\documentclass[aoas]{imsart}

\RequirePackage{amsthm,amsmath,amsfonts,amssymb}
\RequirePackage[authoryear]{natbib}
\RequirePackage[colorlinks,citecolor=blue,urlcolor=blue,hypertexnames=false]{hyperref}
\RequirePackage{graphicx}
\usepackage{booktabs}
\usepackage{multirow}

\startlocaldefs
\theoremstyle{plain}

\newtheorem{theorem}{Theorem}[section]
\newtheorem{lemma}[theorem]{Lemma}
\theoremstyle{definition}
\newtheorem{definition}[theorem]{Definition}

\endlocaldefs

\newtheorem{assumption}{Assumption}

\graphicspath{{Figures/}}
\newcommand{\EE}{\mathbb{E}}  
\newcommand{\NN}{\mathbb{N}}
\newcommand{\PP}{\mathbb{P}}
\newcommand{\RR}{\mathbb{R}}
\newcommand{\T}{\top}

\newcommand{\bel}{\begin{eqnarray}\label}
\newcommand{\eel}{\end{eqnarray}}
\newcommand{\bes}{\begin{eqnarray*}}
\newcommand{\ees}{\end{eqnarray*}}
\newcommand{\bei}{\begin{itemize}}
\newcommand{\beiftnt}{\begin{itemize}\footnotesize}
\newcommand{\eei}{\end{itemize}}

\def\benu{\begin{enumerate}}
\def\eenu{\end{enumerate}}

\def\complex{\mathop{{\rm I}\kern-.58em\hbox{\rm C}}\nolimits}

\def\mathbold{\boldsymbol} 

\def\ba{\mathbold{a}}

\def\bA{\mathbold{A}}

\def\bb{\mathbold{b}}

\def\bB{\mathbold{B}}

\def\bc{\mathbold{c}}

\def\bC{\mathbold{C}}

\def\bH{\mathbold{H}}

\def\bI{\mathbold{I}}

\def\bM{\mathbold{M}}

\def\bO{\mathbold{O}}

\def\bR{\mathbold{R}}

\def\bS{\mathbold{S}}

\def\bU{\mathbold{U}}

\def\bx{\mathbold{x}}

\def\bX{\mathbold{X}}

\def\bY{\mathbold{Y}}

\def\bz{\mathbold{z}}

\def\bZ{\mathbold{Z}}

\def\bGamma{\mathbold{\Gamma}}

\def\bTheta{\mathbold{\Theta}}

\def\lam{\lambda}

\def\bmu{\mathbold{\mu}}

\def\bSigma{\mathbold{\Sigma}}

\makeatletter
\newcommand*{\rom}[1]{\expandafter\@slowromancap\romannumeral #1@}
\makeatother

\begin{document}

\begin{frontmatter}
\title{Learning CNN Filters via Generalized Stein's Method}
\runtitle{Learning CNN Filters via Generalized Stein's Method}

\begin{aug}
\author[A]{\fnms{Guang}~\snm{Yang}\ead[label=e1]{guangyang@cuhk.edu.hk}},
\author[B]{\fnms{Wei}~\snm{Shi}\ead[label=e2]{shiwei1997@connect.hku.hk}},
\author[B]{\fnms{Yuan}~\snm{Cao}\ead[label=e3]{yuancao@hku.hk}}
\and
\author[B]{\fnms{Long}~\snm{Feng}\thanks{Corresponding author.}\ead[label=e4]{lfeng@hku.hk}}
\address[A]{Department of Statistics and Data Science, The Chinese University of Hong Kong\printead[presep={ ,\ }]{e1}}

\address[B]{School of Computing \& Data Science, The University of Hong Kong\printead[presep={,\ }]{e2,e3,e4}}
\end{aug}

\begin{abstract}
Convolutional Neural Networks (CNNs) have undoubtedly revolutionized image data analysis and the field of computer vision. 
As the cornerstone of CNNs, the convolution operation enables the networks to extract abstract features and uncover hidden relationships in the image data.
This paper considers the problem of estimating convolution filters from a statistical perspective using a classical tool --- Stein's formula. We first formulate CNNs into a general index model with matrix-valued input, where convolution filters can be viewed as index vectors. Furthermore, we propose a novel singular value decomposition (SVD) based approach to accurately learn the convolution filters based on a generalized version of the first-order Stein's formula.
Theoretical analysis suggests that our estimation achieves an optimal convergence rate, comparable to that of generalized linear models where the link function is known. Extensive simulation studies and real data analyses demonstrate that our approach outperforms popular deep learning algorithms, such as Adam. Notably, our method extends beyond filter estimation and can be applied to nonlinear dimension reduction, providing a viable pathway for representation learning.

\end{abstract}


\end{frontmatter}

\section{Introduction}\label{sec:intro}

In recent years, machine learning techniques have experienced substantial growth and found widespread application in diverse fields, including computer vision \citep{he2016deep, voulodimos2018deep, chai2021deep}, natural language processing \citep{vaswani2017attention, fathi2018deep, lauriola2022introduction}, and more. 
Convolutional Neural Network (CNN) is undoubtedly one of the most successful frameworks that revolutionized image data analysis and the field of computer vision. The success of CNN can be attributed to its intricate structure design, which incorporates convolution operations, pooling layers, and more. In particular, the convolution operation serves as the cornerstone of CNNs, enabling them to extract abstract features and capture hidden relationships from input data --- a task that was challenging for conventional approaches. In practice, the parameters of CNNs, including convolutional filters, are learned via gradient descent algorithms, which have proven highly effective for large-scale natural image tasks.

However, gradient-based training may encounter challenges in more demanding settings such as medical imaging. First, medical imaging data often have very limited sample size due to reasons such as high acquisition cost or privacy constraints. In contrast, deep learning models typically rely on large amounts of data for effective training, and may become unstable and prone to overfitting when data are scarce \citep{shen2017deep}. Second, deep learning methods have achieved their greatest success in the domain of natural images, but medical images differ fundamentally from natural images. On one hand, the underlying signals of medical images are often subtle and less visually discernible, while the associated biological mechanisms are complex \citep{zhu2023statistical} and require greater interpretability \citep{rudin2019stop}. On the other hand, medical imaging data tend to exhibit more uniform patterns due to standardized acquisition and preprocessing procedures, leading to more structured distributions that may be better exploited by statistical approaches.

This paper is motivated by the problem of learning convolutional filters in CNNs, and proposes a statistical method to estimate filters in the first convolutional layer without training the full neural network, thereby providing an efficient method for learning convolutional representations. We begin by formulating CNNs into a generalized index model with matrix-valued inputs, where convolutional filters act as index vectors and subsequent components like pooling or fully-connected layers are incorporated into the unknown link function. Based on a generalized version of first-order Stein's identity, we develop a novel singular value decomposition (SVD)-based approach to estimate the column space of convolutional filters without requiring knowledge of the link function. Theoretical analysis shows that the proposed estimator achieves an optimal convergence rate, matching that of generalized linear models where the link function is known. Comprehensive simulation studies show that our approach outperforms the widely used Adam algorithm. In real applications on brain MRI data, representations learned by our method can not only improve predictive performance, but also enhance interpretability. These results highlight the practical value of our approach beyond filter estimation, demonstrating its potential for nonlinear dimension reduction and providing a feasible route for representation learning using CNNs. The core contributions of this work are summarized as follows:
\begin{itemize}
	\item We reformulate the first convolutional layer of a CNN as a generalized index model with matrix-valued input and vector-valued indices. This provides a new statistical interpretation of CNN filters.
	\item We derive and use a matrix-valued Stein's identity for this CNN-induced index structure, leading to an SVD-based estimator of the filter column space.
	\item We develop theory for the original estimator, the truncated estimator under heavy-tailed settings, and a practical Gaussian plug-in estimator with unknown covariance.
	\item We demonstrate the method on ADNI brain MRI data, showing its value for representation learning, prediction, and interpretable feature discovery.
\end{itemize}
Coupled with CNNs, our study is conducted under a fairly general framework in several senses: (i) we adopt a generalized index model that extends both single-index and multi-index models; (ii) it accommodates a broad class of distributions, including Gaussian designs and heavy-tailed settings; (iii) without assuming a known distribution, we investigate the practical scenario of Gaussian inputs with unknown covariance, supported by empirical evidence.

\subsection{Related works}
The index model is a classical statistical method which generalizes the linear regression to a nonlinear format from a different perspective than neural networks.
Specifically, a multi-index model (MIM) usually takes the form of $\EE(y|\bx) = f(\bx^\T\bb_1, \ldots, \bx^\T\bb_R)$,
where $\bb_1, \ldots, \bb_R$ are the target index vectors, and $f$ is an unknown multi-variate nonlinear function, also known as the link function. When $R=1$, the MIM reduces to the single index model (SIM), which leads to the index $\bx^\T\bb$ and the function $f$ being univariate. Index models not only offer considerable adaptability by incorporating the nonlinear link function, but also avoid the curse of dimensionality by representing the outcome through a low-dimensional representation of the input.
A classical line of research on index models is closely connected to sufficient dimension reduction, among which sliced inverse regression (SIR) \citep{li1991sliced} is one of the most influential approaches. This framework has inspired a large body of extensions, including sliced average variance estimation \citep[SAVE;][]{cook1991discussion}, and tensor-valued extensions such as TensorSIR \citep{ding2015tensor}. However, these methods typically rely on relatively strong probabilistic assumptions, such as the linearity condition and constant variance condition, which are often satisfied only under elliptical or Gaussian-type distributions.

Another line of research on learning index models is motivated by the classical Stein’s identity \citep{stein1981estimation}, which implies that, in the simplest setting where $\bx$ follows a standard Gaussian distribution, the index vector $\bb$ is proportional to $\EE(\bx \cdot y)$, regardless of the unknown link function $f$. Early works along this direction include average derivative estimation \citep[ADE;][]{stoker1986consistent} and principal Hessian directions \citep[PHD;][]{li1992principal}. More recently, Stein-based methods have been revisited and substantially extended to more general settings. For instance, \citet{plan2016generalized} and \citet{plan2017high} considered SIM with sparse index vectors, and \cite{balasubramanian2025functional} extended the framework to functional regression. \citet{yang2017high} and \citet{goldstein2018structured} generalized the Gaussian input assumption to heavy-tailed or non-Gaussian scenarios, and \citet{fan2023understanding} studied implicit regularization in SIM. Nevertheless, most existing works are developed under relatively simple and idealized settings, typically focusing on SIM with vector-valued inputs and assuming that the input distribution is known. As a result, they provide limited practical guidance for handling complex real-world data such as imaging data.

The rest of the paper is organized as follows. Section \ref{sec:adni} introduces the ADNI dataset in detail. Section \ref{sec:model} establishes the connections between CNNs and index models, and present our SVD-based estimation. Section \ref{sec:theory} provide theoretical analyses for estimating filters under different settings. Section \ref{sec:real} presents the ADNI data analysis, and Section \ref{sec:simu} includes comprehensive simulation studies. Finally, Section \ref{sec:conc} concludes with further discussions.

\textit{Notations.}
We use bold uppercase letters $\bA$, $\bB$ to denote matrices, bold lowercase letters $\ba$, $\bb$ to denote vectors. We let $\text{vec}(\cdot)$ be the vectorization operator and $\text{vec}^{-1}_{(\cdot)}(\cdot)$ be its inverse with the subscripts indicating the matrix size. For a vector ${\bf v}$, $\|{\bf v}\|_q=(\sum_{j} |v_j|^q)^{1/q}$ is the $\ell_q$ norm. For a matrix $\bA$, $\|\bA\|_F=\sqrt{\sum_{i,j} \bA_{i,j}^2}$ is the Frobenius norm, $\lam_{min}(\bA)$ is the minimum eigenvalue of matrix $\bA$. Moreover, we let $\langle\cdot,\cdot\rangle$ to denote the inner product. In addition, given two sequences $\{x_n\}$ and $\{y_n\}$, we denote $x_n = \mathcal{O}(y_n)$ if $|x_n|\le C_1 |y_n|$ for some absolute constant $C_1$ and denote $x_n = \Omega(y_n)$ if $|x_n|\ge C_2 |y_n|$ for some absolute constant $C_2$.

\section{The ADNI dataset}\label{sec:adni}
We analyze brain imaging data from the Alzheimer’s Disease Neuroimaging Initiative (ADNI\footnote{https://adni.loni.usc.edu/}), a study aimed at identifying and monitoring Alzheimer's disease through clinical assessments, genetic analysis, imaging data, and more. Alzheimer’s disease is one of the most common causes of dementia among older adults, and is characterized by a progressive decline in cognitive functions, particularly memory, reasoning, and language abilities. It is associated with structural and functional changes in the brain, including cortical atrophy and the degeneration of specific regions.

The primary goal of this work is to use CNNs to study brain imaging data, with ADNI serving as an illustrative example and cognitive test scores used as response variables. In particular, our study aims to learn effective convolutional representations of brain images that preserve important spatial and pattern information. These learned representations also facilitate dimension reduction, enabling the analysis of high-dimensional imaging data with limited sample sizes. We expect the proposed method to reduce reliance on gradient-based training, thereby lowering computational cost and improving stability, while also enhancing predictive performance and interpretability.

\textbf{Brain MRI.}
For brain imaging data, we focus on magnetic resonance imaging (MRI), as explanatory variables. Each participant in the analysis has a T1-weighted MRI scan. The scans are carefully preprocessed prior to analysis using a standard pipeline that includes spatial adaptive non-local means (SANLM) denoising \citep{manjon2010adaptive}, resampling, bias correction, affine registration, unified segmentation, skull stripping, and cerebellum removal \citep{ashburner2005unified}. This is followed by local intensity correction and spatial normalization into the Montreal Neurological Institute (MNI) atlas space. After preprocessing, each T1-weighted MRI scan is represented as a tensor and resized to $48 \times 60 \times 48$ to improve computational and storage efficiency. As a result, we obtained 1,059 MRI scans from subjects who are cognitively normal (369), have mild cognitive impairment (401), or have Alzheimer’s disease (289). For each scan, we further extract 10 middle coronal slices (slices 25 through 34 out of the 60 coronal slices), which serve as a form of data augmentation derived from a single scan for each subject. This procedure increases the effective sample size, stabilizes model training, and improves the robustness of method comparisons. Finally, the resulting dataset contains 10,590 images of size $48\times 48$, on which we perform 10-fold cross-validation to compare predictive performance.

\textbf{MMSE score.}
We use the Mini-Mental State Examination (MMSE) scores as response variables. The MMSE is a widely used clinical assessment for cognitive impairment, with scores ranging from 0 to 30. There are four cut-off levels of cognitive impairment with corresponding MMSE scores falling within the intervals of [0, 9], [10, 18], [19, 23] and [24, 30], respectively indicating severe, moderate, mild and no cognitive impairment \citep{tombaugh1992mini}. As a cognitive assessment method, MMSE does not confirm any particular disease, but it is widely used as a diagnostic adjunct and provides an informative reference for the assessment of Alzheimer’s disease. In general, cognitively normal individuals tend to have scores close to 30, whereas individuals with cognitive impairment exhibit lower scores. 
Predicting MMSE scores from brain MRI is inherently nontrivial because MMSE reflects cognitive function only indirectly through structural brain changes that are usually subtle in imaging data. In addition, high inter-subject variability and limited sample sizes in neuroimaging further constrain achievable predictive accuracy, making it a more challenging yet important problem to investigate.

We carry out two independent studies using the augmented dataset. In Section \ref{subsec:simu_y}, we consider real MRI images along with simulated responses, where we can simulate the convolution filters and directly evaluate the filter estimation performance.
In Section \ref{subsec:mmse}, we incorporate real MMSE scores as responses to showcase the prediction performance based on representations from our method under various downstream predictive models.

\section{CNN, Index Model, and their connections}\label{sec:model}
\subsection{CNN and Index Model}
Let $\bX^{\text{ori}}\in\mathbb{R}^{P_1\times P_2}$ denote the original matrix-valued input, such as raw image data. In a typical convolutional neural network, the raw image usually first goes through a convolutional layer, and then proceeds to pooling and fully connected layers. We shall focus on the first convolutional layer and denote $\bB_1, \ldots, \bB_R \in \mathbb{R}^{d_1\times d_2}$ as $R$ convolution filters in the first layer of a CNN. 
Then, the output of a convolutional neural network can be written as the following form
\begin{align}\label{model:cnn}
	G(\bX^{\text{ori}}) = g(\sigma(\bX^{\text{ori}} \star \bB_1), \ldots, \sigma(\bX^{\text{ori}} \star \bB_R)).
\end{align}
Here $\star$ is a convolution operator to be discussed in detail later, $\sigma(\cdot)$ is a component-wise 
nonlinear activation function, such as ReLU (Rectified Linear Unit) or Sigmoid function. Furthermore, the multivariate function $g$ incorporates the remaining structure of CNN, including pooling layers, fully-connected layers, and other components in a CNN. If we further combine $g$ with the activation function $\sigma$, the CNN output (\ref{model:cnn}) can be written as
\begin{align}\label{model:ori_multi}
	G(\bX^{\text{ori}}) = f(\bX^{\text{ori}}\star\bB_1, \ldots, \bX^{\text{ori}}\star\bB_R).
\end{align}

We now turn our attention to the convolution operator $\star$. We shall first consider a non-overlapping convolution, further extensions to its overlapping version will be discussed later. Given any matrix $\bM\in \mathbb{R}^{P_1\times P_2}$ and a filter $\bB\in \mathbb{R}^{d_1\times d_2}$, where $d_1, d_2$ are factors of $P_1, P_2$ respectively, the non-overlapping convolution is defined as
\begin{equation}
	\bM\star\bB\in \mathbb{R}^{p_1\times p_2}, \ \ (\bM\star\bB)_{j,k} = \left\langle\bM_{j, k}^{d_1, d_2}, \bB \right\rangle.
\end{equation}
Here, $(p_1, p_2) = (P_1/d_1, P_2/d_2)$ and $\bM_{j,k}^{d_1, d_2}$ represents the $(j,k)$-th block of size $d_1\times d_2$ in $\bM$ for $1\le j\le p_1$ and $\ 1\le k\le p_2$. With deep learning language,  the stride size is set equal to the filter size under non-overlapping scenario. The convolution operation aims to extract feature information, thereby enhancing prediction in the subsequent layers. 
While through the non-overlapping design, dimension reduction can be achieved. 
Notably, the non-overlapping convolution, being the simplest case, has been extensively studied in the literature, for example, \citet{brutzkus2017globally,du2018gradient,cao2019tight,feng2024deep}.
Let $p=p_1p_2$ be the total number of non-overlapping blocks and $d=d_1d_2$ be the number of elements in each block.

Now we consider a matrix reshaping operator that allows us to further connect the CNN and index models. Given any matrix $\bM\in\mathbb{R}^{P_1\times P_2}$, the operator $\mathcal{R}_{(d_1, d_2)}: \mathbb{R}^{P_1\times P_2}\rightarrow \mathbb{R}^{p\times d}$ is defined as a mapping from $\bM$ to
\begin{align}\label{Rmat}
	\mathcal{R}_{(d_1, d_2)}(\bM) = \left[\text{vec}\left(\bM_{1,1}^{d_1, d_2}\right), \ldots, \text{vec}\left(\bM_{p_1,p_2}^{d_1, d_2}\right)\right]^\T.
\end{align}
A key property holds for $\bX^{\text{ori}}\in \mathbb{R}^{P_1\times P_2}$ and $\bB \in \mathbb{R}^{d_1\times d_2}$ that
$$\text{vec}(\bX^{\text{ori}}\star\bB) = \mathcal{R}_{(d_1, d_2)}(\bX^{\text{ori}})\text{vec}(\bB).$$
Figure \ref{fig:R_opt} illustrates the transformation of convolution operation into matrix multiplication through the operator $\mathcal{R}_{(d_1, d_2)}(\cdot)$.

\begin{figure}[!htbp]
	\centering
	\includegraphics[width=\textwidth]{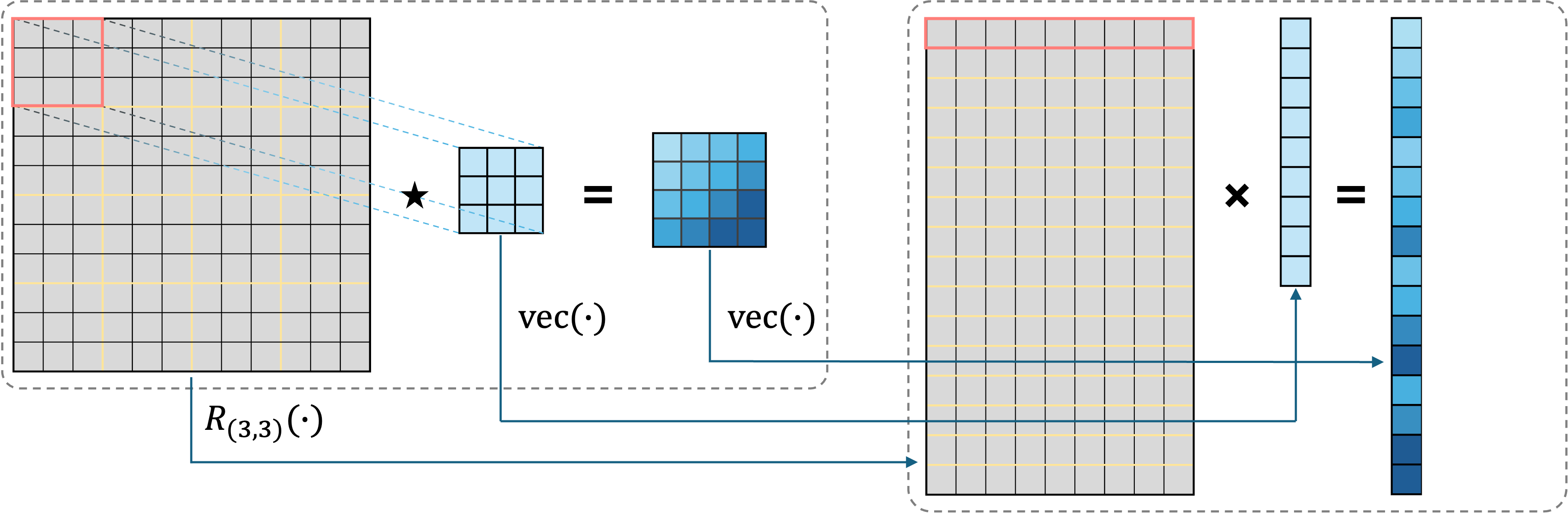}
	\caption{Transformation of convolution operation into matrix multiplication through the operator.}
	\label{fig:R_opt}
\end{figure}

Let $\bX=\mathcal{R}_{(d_1, d_2)}(\bX^{\text{ori}}) \in \mathbb{R}^{p\times d} $ be the reshaped input and $\bb_r = \text{vec}(\bB_r) \in \mathbb{R}^{d}$ for $r=1, \ldots, R$. Then we can rewrite (\ref{model:ori_multi}) as
\begin{align}\label{model:mim_multi}
	G(\bX^{\text{ori}}) = f(\bX\bb_1, \ldots, \bX\bb_R).
\end{align}
With a slight abuse of notation, we also use the same function notation $f$ in (\ref{model:mim_multi}). However, we shall note that the $f$ in (\ref{model:mim_multi}) is not identical to that in (\ref{model:ori_multi}) as the inputs were reshaped.
If we model certain outcome $Y$ through the above formulation, we have
\begin{equation}\label{eq7}
	\EE(Y|\bX)=f(\underbrace{\bX\bb_1, \ldots, \bX\bb_R}_{p\times R}).
\end{equation}
Model~\eqref{eq7} clearly bears a resemblance to an MIM, although in fact they are different. In a standard MIM, the input is typically a vector, so each index is a scalar. In model~\eqref{eq7}, however, $\bX$ is a matrix of dimension $p\times d$, which makes each index $\bX\bb_r$ a $p$-dimensional vector. Consequently, the function in \eqref{eq7} is a more general function with an input dimension $p\times R$. Let $\bTheta = \left(\bb_1, \ldots, \bb_R\right) \in \mathbb{R}^{d\times R}$ denote the collection of convolutional filters. Then model~\eqref{eq7} can be compactly written as $\EE(Y|\bX)=f(\bX\bTheta)$.

\subsection{CNN filter estimation via Stein's method and SVD}\label{sec:proposed_method}
Based on model~\eqref{eq7}, we assume the relationship between $Y$ and $\bX$ is $Y=f(\bX\bTheta)+\epsilon$, where the noise term $\epsilon$ is independent of $\bX$. We assume that $\epsilon$ is centered and sub-Gaussian, i.e. $\EE[\epsilon]=0$ and $\EE[|\epsilon|^{k}]^{\frac{1}{k}} \leq \sigma\sqrt{k}$ for all $k \in \NN$ and an absolute constant $\sigma>0$. Let $\bX_i=\mathcal{R}_{(d_1, d_2)}(\bX_i^{\text{ori}})\in\mathbb{R}^{p\times d}$ and $Y_i\in\mathbb{R}$ for $i\in [n]$ be 
$n$ i.i.d. observations that follow $Y_i=f(\bX_i\bTheta)+\epsilon_i$. Similar to index models, $\bTheta$ is not identifiable when $f$ is unknown. Specifically, it is identifiable only up to an invertible linear transformation, since any such transformation can be absorbed into $f$. Therefore, our goal is to learn the column space in $\bTheta$ based on a generalized version of Stein's identity. To get started, we shall first define the score function associated with the input $\bX$.
\begin{definition}\label{def:score}
	Let $\bX\in\RR^{p\times d}$ be a random matrix with density $P:\RR^{p\times d}\rightarrow\RR$. The score function $S(\bX):\RR^{p\times d}\rightarrow\RR^{p\times d}$ associated with $\bX$ is defined as
	\begin{align*}
		S(\bX):=-\nabla_{\bX} [\log P(\bX)]=-\nabla_{\bX} P(\bX)/P(\bX).
	\end{align*}
	For any $\{j,k\}\in[p]\times[d]$, $S_{jk}(\bX)$ is defined as the $(j,k)$-th element of $S(\bX)$.
\end{definition}

Then we have the following generalized version of first-order Stein's identity.
\begin{lemma}\label{lemma:CNNexpectation}
	Let $f: \mathbb{R}^{p\times R}\mapsto \mathbb{R}$ 
	be the link function. Let $\bZ=\bX\bTheta$. Suppose that the expectations $\EE[YS(\bX)]$ and $\EE\left[\nabla_{\bZ} f (\bX\bTheta)\right]$ both exist and are well-defined. 
	Then we have
	\begin{align}\label{eq:main}
		\EE[YS(\bX)] = \EE\left[\nabla_{\bZ} f (\bX\bTheta)\right] \cdot \bTheta^{\top}.
	\end{align}
\end{lemma}

Lemma \ref{lemma:CNNexpectation} serves as the basis for our estimation of $\bTheta$. It shows that provided $\EE[YS(\bX)]$ has rank $R$, the column space of $\bTheta$ can be determined by the top $R$ right singular vectors of $\EE[YS(\bX)]$. With the sample version $(1/n)\sum_{i=1}^{n}Y_iS(\bX_i)$ being a natural estimate for $\EE[YS(\bX)]$, we propose the following SVD based estimation for $\bTheta$,
\begin{align} 
	\hat{\bTheta}=\mathrm{SVD}_{v,R}\left(\frac{1}{n}\sum_{i=1}^{n}Y_iS(\bX_i)\right), \label{eq:betahat}
\end{align}
where $\mathrm{SVD}_{v,R}(\bM)$ refers to the top $R$ right singular vectors of a matrix $\bM$.

We emphasize that this method is general as it applies to arbitrary distribution of input $\bX$, provided that the score $S(\bX)$ is known or well estimated. Moreover, this estimator does not rely on the link function $f$, therefore regardless of the complexity of subsequent pooling or fully connected layers in a CNN. Based on SVD, the proposed estimator aims to learn a compact set of informative convolutional directions rather than to reproduce every possibly redundant filter in a large CNN. This is consistent with empirical observations that CNNs often learn filters with substantial redundancy \citep{jaderberg2014speeding, wang2020orthogonal}. Therefore, the orthogonality of $\hat{\bTheta}$ should be regarded as a basis-recovery feature of the estimator, rather than as a structural assumption on all practical CNN filters.

Although the method is introduced to estimate filters in the first convolutional layer, in practice it serves as a representation learning method that efficiently extracts convolutional features without training the full CNN. The resulting convolutional features can then be used in downstream tasks and integrated with either statistical models or neural networks. Specifically, an applied researcher can use this method following steps below.
\begin{itemize}
	\item[\textbf{Step 1.}] \textbf{Estimate the score function.} Under a suitable parametric model for the input distribution, the score can be estimated by plug-in estimators. For example, under a Gaussian model, the empirical mean and covariance give an explicit estimate of the score function, as described in Section~\ref{sec:emp}.
	\item[\textbf{Step 2.}] \textbf{Estimate filters.} Compute $(1/n)\sum_{i=1}^nY_i\hat S(\bX_i)$ and perform SVD. The top right singular vectors form $\hat{\bTheta}$, and each column is reshaped into a convolutional filter $\hat \bB_r=\mathrm{vec}^{-1}_{(d_1,d_2)}(\hat{\bTheta}_{[:,r]})$.
	\item[\textbf{Step 3.}] \textbf{Construct representations and fit downstream models.} Apply the estimated filters to the original images to obtain convolutional feature maps, and then use these features in a downstream model such as linear regression, LASSO, FCN, or CNN.
\end{itemize}
A notable advantage of the proposed method is its seamless integration with downstream statistical models, as the extracted features can be directly used as their inputs. This enables a natural combination of the representation power of CNN-based feature extraction with the strengths of classical statistical methods, including interpretability and inference.
We also note that if the input distribution is completely unknown and no suitable parametric model is available, score estimation becomes a substantive problem. In such cases, nonparametric score estimation methods, such as \citet{zhou2020nonparametric}, may provide an alternative, but their accuracy should be checked for the application at hand.

\subsection{The truncated form}\label{sec:truncated_form}
Under heavy-tailed data, the straightforward sample version $(1/n)\sum_{i=1}^{n}Y_iS(\bX_i)$ may not be an optimal estimate for $\EE[YS(\bX)]$ due to a lack of concentration. In fact, not only can the input $\bX$ be heavy-tailed, but the outcome $Y$ or the score function $S(\bX)$ could also exhibit heavy-tail behavior due to the unknown link function $f$. These issues have been discussed in the literature, where the truncation argument has been proposed, as seen in works such as \citet{catoni2012challenging,minsker2018sub,yang2017high}. We adopt similar strategies and propose an improved estimator of the truncated form.

We begin by introducing a class of non-decreasing functions $\phi:\RR\rightarrow\RR$ satisfying
\begin{align*}
	-\log(1-x+x^2/2)\leq\phi(x)\leq\log(1+x+x^2/2), \ \ x\in\RR.
\end{align*}
A straightforward function satisfying the upper and lower exponential-type bounds is
\begin{align}\label{phi}
	\phi(x)=\left\{\begin{array}{ll}
		-\log(1-x+x^2/2),& x <0,\\
		\log(1+x+x^2/2), & x\ge 0.
	\end{array}
	\right.
\end{align}
It is a convenient choice as it attains these bounds tightly and modifies the original estimator only as much as needed to control extreme values, while other truncation functions with the same bounding properties could also be used. Figure \ref{fig:phi} illustrates the function $\phi$.
\begin{figure}[!htbp]
	\centering
	\includegraphics[width=\linewidth]{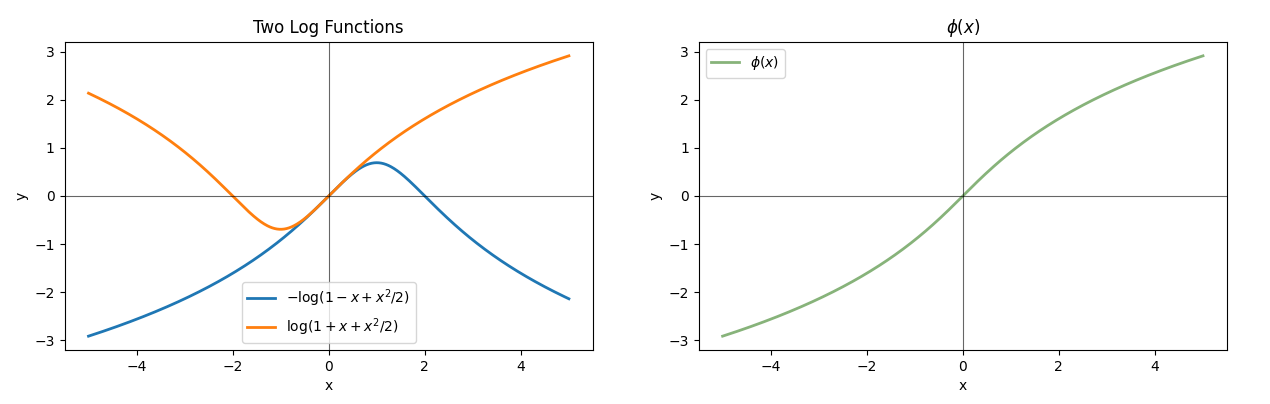}
	\caption{The plot of function $\phi$.}
	\label{fig:phi}
\end{figure}

We introduce a linear mapping $\psi: \RR^{p\times d}\rightarrow\RR ^{p\times d}$ based on the non-decreasing function $\phi$.
For a given matrix $\bX\in\RR^{p \times d}$, consider the spectral decomposition of its Hermitian dilation:
\begin{align*}
	\bX^* = \begin{bmatrix} 0 & \bX\\ \bX^{\top} & 0 \end{bmatrix} = \bU\bSigma \bU^{\top} .
\end{align*}
Further, define $\Tilde{\bX}=\bU(\phi \circ \bSigma)\bU^{\top}$ and write it into block form as 
\begin{align*}
	\Tilde{\bX}=\bU(\phi \circ \bSigma)\bU^{\top}=\begin{bmatrix} \underbrace{\Tilde{\bX}_{11}}_{p\times p} & \underbrace{\Tilde{\bX}_{12}}_{p\times d}\\ \underbrace{\Tilde{\bX}_{21}}_{d\times p}& \underbrace{\Tilde{\bX}_{22} }_{d\times d}\end{bmatrix},
\end{align*}
where $\phi \circ \bSigma$ indicates that the function $\phi$ is applied element-wise on $\bSigma$. Then we define $\psi(\bX)=\tilde{\bX}_{12}$. Based on $\psi$, we introduce a new estimator  $\Tilde{\bTheta}$ for $\bTheta$ with a truncated form 
\begin{align}\label{theta_tilde}
	\Tilde{\bTheta}=\mathrm{SVD}_{v,R}\left(\frac{1}{n}\sum_{i=1}^{n}\tau\left(Y_iS(\bX_i)\right)\right),
\end{align}
where $\tau\left(Y_iS(\bX_i)\right)=1/\theta\cdot\psi\left[\theta Y_iS(\bX_i)\right]$, and $\theta>0$ is a thresholding parameter. In later analysis, we will provide theoretical analyses of $\Tilde{\bTheta}$ and demonstrate its advantages over $\hat{\bTheta}$ in (\ref{eq:betahat}) under heavy-tailed distributions.

\section{Theoretical guarantees for filter estimation}\label{sec:theory}
\subsection{Theoretical guarantees for single filter estimation}\label{sec:single_filter}
In this section, we provide theoretical guarantees for estimating $\bTheta$ with a single convolution filter. With $R=1$, we let $\bb=\text{vec}(\bB)=\bTheta\in\mathbb{R}^d$ be the target filter and denote the original and truncated version of the estimated filter as $\hat{\bb}=\hat{\bTheta}$ and $\tilde{\bb}=\tilde{\bTheta}$, respectively.
We assume that $\|\bb\|_2=1$ to avoid identifiability issues.
We first consider the original estimator $\hat{\bb}$ under sub-Gaussian input. To get started, we present necessary conditions below.

\begin{assumption}\label{assumption:sub1}
	Suppose that $\bX_i\in\mathbb{R}^{p\times d}$, $i\in [n]$ are i.i.d. observations.
	For any $i\in [n]$, $j\in[p]$ and $k\in[d]$, assume that the score $S_{jk}(\bX_i)$ is a sub-Gaussian random variable, i.e. $(\EE|S_{jk}(\bX_i)|^q)^\frac{1}{q}\leq L\sqrt{q}$ for $q\geq 1$, where $L$ is a positive constant.
\end{assumption}

Assumption~\ref{assumption:sub1} requires that each entry of the score function $S(\bX)$ is sub-Gaussian. Such an assumption is rather mild and can be satisfied by a wide range of random distributions.
For example, when $\bX$ follows a Gaussian distribution with entries $X_{jk}\sim N(\mu_{jk},\sigma_{jk}^2)$, it results in $S_{jk}(\bX)=(X_{jk}-\mu_{jk})/\sigma_{jk}^2$ also being a Gaussian random variable, which is clearly sub-Gaussian. Moreover, when $X_{jk}$ follows a heavier tailed Students' $t$ distribution, it can be easily verified that the score is bounded and inherently sub-Gaussian. This implies that even if the input $\bX$ is not sub-Gaussian, the score could still potentially be sub-Gaussian. 

\begin{assumption}\label{assumption:sub2}
	Assume that there exists an absolute constant $T>0$ such that $\EE[f(\bX\bTheta)^{2q}]^{\frac{1}{2q}} \leq T\sqrt{q}$ for all $q\in\NN_+$.
\end{assumption}

Assumption~\ref{assumption:sub2} requires that $f(\bX\bTheta)$ is a sub-Gaussian random variable. This assumption is relatively mild. Clearly, Assumption~\ref{assumption:sub2} automatically holds when $f\in \mathbb{R}^{p\times R}\mapsto \mathbb{R}$ is a bounded function. Moreover, when $f$ is Lipschitz continuous and $\bX$ is sub-Gaussian, it is easy to show that Assumption~\ref{assumption:sub2} could also be satisfied.

\begin{assumption}\label{assumption:sub3}
	Let $\bZ=\bX\bb$. For any $j \in [p]$, assume that $\EE\big[\nabla_{Z_j} f(\bZ)\big] = \Omega(1)$.
\end{assumption}

In Assumption~\ref{assumption:sub3}, we assume that the derivative of the non-linear function $f$ at $\bZ=\bX\bb$ maintains a constant level. This assumption ensures identifiability by requiring that the function value $f(\bX\bb)$ is adequately sensitive to perturbations around the true value $\bb$.
This assumption is mild, as it merely assumes a constant scale of the derivatives in expectation.
In practice, the assumption can be evaluated by inspecting the empirical singular value spectrum of the matrix $\frac{1}{n}\sum_{i=1}^{n}Y_iS(\bX_i)$, which is supported when the leading singular value is significantly larger than the remaining spectrum, illustrated in the Supplementary Material.

\begin{theorem}\label{thm:subbeta}
	Suppose the model $\EE(Y_i|\bX_i)=f(\bX_i\bb)$ for $i\in [n]$.
	Under Assumption~\ref{assumption:sub1} to \ref{assumption:sub3}, we have with probability at least $1-\delta$ that
	\begin{align}\label{thm-3-4-1}
		\|\hat{\bb}-\bb\|_2 = \mathcal{O}\left(\sqrt{\frac{d}{n}\cdot\log\frac{pd}{\delta}}\right).
	\end{align}
\end{theorem}

Theorem~\ref{thm:subbeta} suggests that the $\ell_2$-convergence rate of $\hat{\bb}$ is $\mathcal{O}\left(\sqrt{d\log(pd)/n}\right)$ under the sub-Gaussian setting. 
This rate is optimal up to a logarithmic factor. In fact, consider the simplest scenario where $f$ is linear, meaning that $\EE(Y_i|\bX_i)=\bc^\T\bX_i\bb$, with $\bc\in\mathbb{R}^p$ being a known vector. In this case, the $\ell_2$-convergence rate of an OLS estimate of $\bb$, denoted as $\hat{\bb}^{ols}$, is $\mathcal{O}\left(\sqrt{d/n}\right)$. Our rate in Theorem~\ref{thm:subbeta} matches with this OLS rate apart from a logarithmic factor, which confirms the optimality of (\ref{thm-3-4-1}).

Now we focus on the truncated estimator, for which we can provide sharp error bounds even when both the outcomes and the scores exhibit heavy-tailed behavior.

\begin{assumption}\label{assumption:trunc1}
	Suppose the model $\EE(Y_i|\bX_i)=f(\bX_i\bb)$ for $i\in [n]$. Assume that there exists an absolute constant $M>0$ such that $\EE(Y_i^4)\leq M$ and $\EE(S_{jk}^4(\bX_i))\leq M$ for all $i\in [n]$, $j\in[p]$ and $k\in[d]$.
\end{assumption}

Assumption~\ref{assumption:trunc1} relaxed the sub-Gaussian Assumption~\ref{assumption:sub1} to a bounded fourth-moment assumption on the outcome $Y$ and the score $S(\bX)$. Under this assumption, heavy-tailed inputs could be included.  
We note that similar assumptions have been studied in the literature, such as \citet{yang2017high,fan2018large,fan2021robust,fan2023understanding}. With the bounded moment assumption, we can present the theoretical results for the truncated estimator.

\begin{theorem}\label{thm:truncbeta}
	Suppose the model $\EE(Y_i|\bX_i)=f(\bX_i\bb)$ for $i\in [n]$. Suppose that Assumption \ref{assumption:sub3} holds. Further assume Assumption \ref{assumption:trunc1} with constant $M$.
	Let $\theta=\sqrt{2\log(2(p+d)/\delta)/(nMpd)}$. Then, with probability at least $1-\delta$, we have
	\begin{align}\label{thm-3-6-1}
		\|\Tilde{\bb}-\bb\|_2 = \mathcal{O}\left(\sqrt{\frac{d}{n}\cdot\log\frac{p+d}{\delta}}\right).
	\end{align}
\end{theorem}

By Theorem \ref{thm:truncbeta}, the truncated estimator could achieve 
a convergence rate of $\mathcal{O}\left(\sqrt{d\log(p+d)/n}\right)$ with a carefully chosen truncation parameter $\theta$. In other words, even if both the outcome $Y_i$ and the score $S(\bX_i)$ are heavy-tailed, a nearly optimal convergence rate could still be guaranteed under a bounded moment condition. On the other hand, it is worth noting that the bounded expectation assumption on $f$, referred to as Assumption \ref{assumption:sub2}, is not necessary in Theorem \ref{thm:truncbeta}. In other words, the convergence can be guaranteed solely under the constant derivatives Assumption \ref{assumption:sub3} and the bounded scores assumption \ref{assumption:trunc1}.

\subsection{Generalizations to multiple filters estimation}\label{sec:multiple_filter}
In this section, we consider general cases with $R\ge 2$ filters in CNNs. Similar to Section \ref{sec:single_filter}, we study both the original and truncated versions of the estimator under sub-Gaussian and moment conditions, respectively. However, as previously discussed, when $R\ge 2$, $\bTheta$ is not identifiable with an unknown link function $f$. Thus, our goal is to learn the column space of $\bTheta$ and prove the convergence of the estimators using the following distance measure
\begin{align}\label{dist}
	\mathrm{dist}\left(\bTheta, \hat{\bTheta}\right)=\inf_{\bH\in\mathbb{H}_R}\left\|\bTheta-\hat{\bTheta}\bH\right\|_F,
\end{align}
where $\mathbb{H}_R$ is the set of $R\times R$ orthogonal matrices. It is worth noting that when $R=1$, the column space distance $\text{dist}\left(\bb, \hat{\bb}\right)$ reduces to the vector distance $\|\hat{\bb}-\bb\|_2$ when $\|\bb\|_2=\|\hat{\bb}\|_2$.

In a multi-filter setting, we require assumptions similar to those used for single-filter estimation, with Assumption \ref{assumption:sub3} replaced by Assumption \ref{assumption:multiple1} below.

\begin{assumption}\label{assumption:multiple1}
	Suppose that $(\bX_i,Y_i)$, $i\in [n]$ are i.i.d. observations.
	Let $\sigma_R$ be the $R$-th largest singular value of $\EE[Y_iS(\bX_i)]$. Assume that there exists an absolute constant $C$ such that $\sigma_R \geq C\sqrt{p/R}$.
\end{assumption}

Assumption~\ref{assumption:multiple1} imposes a lower bound of the non-zero singular values of $\EE[Y_iS(\bX_i)]$. This assumption can be viewed as an extension of Assumption \ref{assumption:sub3} in a multiple filters scenario. For example, when each entry of $\EE[\nabla_{\bZ} f (\bX\bTheta)]$ is $\Omega(1/\sqrt{R})$ and the columns of $\EE[\nabla_{\bZ} f (\bX\bTheta)]$ are nearly orthogonal, it is easy to show that all the $R$ non-zero singular values of $\EE[YS(\bX)]$ have a lower bound in scale  of $\Omega(\sqrt{p/R})$. Similarly, in practice it can be also assessed by inspecting the empirical singular value spectrum of $\frac{1}{n}\sum_{i=1}^{n}Y_iS(\bX_i)$, which is supported when the leading $R$ singular values are clearly separated from the remaining spectrum.

\begin{theorem}\label{thm:multiple_sub}
	Suppose the model $\EE(Y_i|\bX_i)=f(\bX_i\bTheta)$ for $i\in [n]$. Under Assumptioin~\ref{assumption:sub1}, Assumption~\ref{assumption:sub2} and Assumption~\ref{assumption:multiple1}, with probability at least $1-\delta$, we have
	\begin{align*}
		\mathrm{dist}\left(\bTheta, \hat{\bTheta}\right)=\mathcal{O}\left(\sqrt{\frac{Rd}{n}\cdot\log\frac{pd}{\delta}}\right).
	\end{align*}
\end{theorem}

In this multiple filters case, we achieve a convergence rate of \(\mathcal{O}\left(\sqrt{Rd\log(pd)/n}\right)\) under the sub-Gaussian assumption. By replacing Assumption \ref{assumption:sub3} with an SVD-type Assumption \ref{assumption:multiple1}, Theorem \ref{thm:multiple_sub} extends the near-optimal convergence rate stated in Theorem~\ref{thm:subbeta} to a multiple filter scenario. Moreover, we can eliminate the sub-Gaussian assumptions as detailed below.

We now consider the truncated estimator without imposing the sub-Gaussian assumption. The following theorem extends Theorem \ref{thm:truncbeta} to the multiple-filter case.

\begin{theorem}\label{thm:multiple_trunc}
	Suppose the model $\EE(Y_i|\bX_i)=f(\bX_i\bTheta)$ for $i \in [n]$. 
	Suppose that Assumption \ref{assumption:multiple1} holds. Further assume Assumption \ref{assumption:trunc1} with constant $M$.
	Let $\theta=\sqrt{2\log(2(p+d)/\delta)/(nMpd)}$.
	Then, with probability at least $1-\delta$, we have
	\begin{align*}
		\mathrm{dist}\left(\bTheta, \tilde{\bTheta}\right) = \mathcal{O}\left(\sqrt{\frac{Rd}{n}\cdot\log\frac{p+d}{\delta}}\right).
	\end{align*}
\end{theorem}

Theorem~\ref{thm:multiple_trunc} characterizes the convergence rate of the column space of $\Tilde{\bTheta}$ for multiple filters under the general setting.
Compared to Theorem~\ref{thm:multiple_sub}, Theorem~\ref{thm:multiple_trunc}  further demonstrates that a sharp convergence rate can be guaranteed 
even when both the outcomes and the scores
exhibit heavy-tailed behavior.

\subsection{Gaussian distribution with unknown Covariances}\label{sec:emp}
To estimate the convolution filters using our approach, a crucial step involves computing the score function, which requires the knowledge of the input distribution.
Nevertheless, the input distributions are usually unknown in practice. When the input distribution has a specific parametric form, a natural solution is to use the plug-in estimator. This involves estimating the unknown parameters from the samples and then plugging them into the scores. In this section, we take Gaussian input as an example and demonstrate the convergence of the plug-in estimators.

Let $\bx_i = \text{vec}(\bX_i)=\text{vec}\left(\mathcal{R}(\bX_i^{\text{ori}})\right)$ be the vectorized image.  Assume that $\bx_i \sim \mathcal{N}(\bmu, \bSigma)$ follows Gaussian distribution with unknown mean $\bmu$ and covariance matrix $\bSigma$. Under such a Gaussian case, it is easy to show that the score function reduces to $S(\bX_i)=\text{vec}_{(p,d)}^{-1}(\bSigma^{-1}\bx_i)$, where $\text{vec}^{-1}_{(p,d)}(\cdot)$ is the inverse operation of transforming a matrix to a vector. Denote the sample mean and sample covariance matrix as $\hat{\bmu}=(1/n)\sum_{i=1}^n\bx_i$ and $\hat{\bSigma}=(1/n)\sum_{i=1}^n\{\bx_i - \hat{\bmu}\}\{\bx_i - \hat{\bmu}\}^\T$, respectively.
By plugging the sample mean and sample covariance matrix into the score, we obtain the following estimator:
\begin{equation}\label{Acheck}
	\check{\bTheta} = \text{SVD}_{v, R}\left(\check{\bA}\right), \ \check{\bA} = \text{vec}^{-1}_{(p, d)}\left(\frac{1}{n}\sum_{i=1}^n Y_i\hat{\bSigma}^{-1}\left(\bx_i - \hat{\bmu}\right)\right).
\end{equation}

\begin{assumption}\label{assumption:taylor}
	Suppose that the link function admits a first-order Taylor expansion of the form
	$f(\bZ)=f(\boldsymbol{0}_{p\times R})+\langle\nabla f(\boldsymbol{0}_{p\times R}), \bZ\rangle + h(\bZ)\|\bZ\|_2^2$ with $\lim_{\|\bZ\|_2\rightarrow \boldsymbol{0}} h(\bZ)=0$. 
	Assume that the remainder term is bounded, $h(\bX_i\bTheta) \le H $ for $i \in [n]$ for some constant $H$.
\end{assumption}

We now present the following theorem which provides a bound on the column space between $\check{\bTheta}$ and the original estimator $\hat{\bTheta}$, where the true mean and covariance are used.
\begin{theorem}\label{thm:plug-in_beta} 
	Suppose the model $\EE(Y_i|\bX_i)=f(\bX_i\bTheta)$ for $i \in [n]$. Assume that $\bx_i=\text{vec}(\bX_i) \sim \mathcal{N}(\boldsymbol{0}, \bSigma)$ with $\lambda_{\min}(\bSigma)>0$. Let $\hat{\bTheta}$ be the original estimator (\ref{eq:betahat}).
	Then under Assumption \ref{assumption:taylor}, we have the following holds with probability approaching 1, 
	\begin{align}
		\mathrm{dist}\left(\check{\bTheta}, \hat{\bTheta}\right) = \mathcal{O}\left(\sqrt{\frac{Rd}{n}}\right).
	\end{align}
\end{theorem}
Theorem \ref{thm:plug-in_beta} shows that column space of the plug-in estimator converges to that of the original estimator at the rate $\mathcal{O}\left(\sqrt{Rd/n}\right)$. In fact, under certain cases, the plug-in estimator may even outperform the original one. To illustrate, consider a natural case where $\bX_i$ follows a zero-mean Gaussian distribution and $f$ is linear, so that $Y_i=\langle\bC,\bX_i\bTheta\rangle+\varepsilon_i$, where $\bC\in\mathbb{R}^{p\times R}$ is an unknown matrix. This model can be viewed as a low-rank trace regression of the form $\EE(Y_i|\bX_i)=\text{tr}(\bA\bX_i^\T)$, where $\bA=\bC\bTheta^\T\in\mathbb{R}^{p\times d}$ is a low-rank matrix. Under this setting, $\check{\bA}$ in (\ref{Acheck}) reduces to the OLS estimate of $\bA$:
\begin{align*}
	\check{\bA}=\text{vec}_{(p,d)}^{-1}\left(\frac{1}{n}\sum_{i=1}^n\hat{\bSigma}^{-1}\bx_iY_i\right)=\bA+\text{vec}_{(p,d)}^{-1}\left(\frac{1}{n}\sum_{i=1}^n\hat{\bSigma}^{-1}\bx_i\varepsilon_i\right).
\end{align*}
As a comparison, the original estimator $\hat{\bA}$ with true mean and covariance becomes:
\begin{align*}
	\hat{\bA}=\text{vec}_{(p,d)}^{-1}\left(\frac{1}{n}\sum_{i=1}^n\bSigma^{-1}\bx_iY_i\right)=\text{vec}^{-1}_{(p, d)}\left(\bSigma^{-1}\hat{\bSigma} \text{vec}\left(\bA\right) \right)+\text{vec}_{(p,d)}^{-1}\left(\frac{1}{n}\sum_{i=1}^n\bSigma^{-1}\bx_i\varepsilon_i\right).  
\end{align*}
By above derivation, $\check{\bA}$ could potentially provide a better estimation of $\bA$ in comparison to $\hat{\bA}$. Particularly, when sample size $n$ is small, 
$\bSigma^{-1}\hat{\bSigma}$ might deviate from $\bI$ and then the estimation performance of $\hat{\bA}$ could be impacted. 
Consequently, the plug-in estimator $\check{\bTheta}$, which consists of the singular vectors of $\check{\bA}$, might yield a smaller estimation error. This finding is further validated by an extensive simulation study presented in Section \ref{sec:exp:emp}.

\section{The ADNI data analysis}\label{sec:real}
In this section, we evaluate the proposed estimators using brain MRI data from the ADNI dataset, which has been introduced in Section~\ref{sec:adni}.
We first consider simulated responses in Section \ref{subsec:simu_y}, where we can directly evaluate the filter estimation performance. We then analyze the real MMSE scores in Section \ref{subsec:mmse}, where the quality of the learned convolutional representations is assessed through downstream prediction performance.

\subsection{Simulated response}\label{subsec:simu_y}
We regard MRI images as input $\bX_i^{\text{ori}}$ and generate the responses $Y_i$ according to the following models
\begin{itemize}
	\item [(i)] Single filter: $Y_i = f\left(\bX_i^{\text{ori}} \star \bB\right) + \varepsilon_i$.
	\item [(ii)] Multiple filters with $R=3$: $Y_i = f\left(\bX_i^{\text{ori}} \star \bB_1, \ \bX_i^{\text{ori}} \star \bB_2, \ \bX_i^{\text{ori}} \star \bB_3\right) + \varepsilon_i$.
\end{itemize}
We consider four kinds of link functions $f$ described as (I) to (IV) in Section \ref{sec:simu}, including linear, nonlinear, fully connected network (FCN) and convolutional neural network (CNN). The filters ($\bB$; $\bB_1$, $\bB_2$, $\bB_3$) are constructed as the singular vectors of a random matrix followed by reshaped into square matrices of size $4\times 4$, which are fixed throughout all experiments. Consequently, the reshaped images $\bX_i=\mathcal{R}_{(4,4)}(\bX^{\text{ori}})\in \RR^{144\times 16}$. The noise term $\varepsilon_i$ is generated as Gaussian distribution $\mathcal{N}(0, \sigma^2)$ with $\sigma=0.1$.

To validate our approach in real data analysis, a key step involves understanding the distributions of the input. We follow the work of \citet{lindquist2008statistical} and assume that $\bX_i^{\text{ori}}$ follows a Gaussian distribution with unknown mean and covariance. We implement the plug-in estimator of our approach, and compare with CNNs trained using Adam \citep{kingma2014adam}, one of the most widely used gradient descent algorithms for training neural networks. We reiterate that our approach does not require the knowledge of link functions while applying Adam necessitates specifying a neural network architecture. To approximate the link functions in a consistent manner, we employ a two-layer fully connected network. Moreover, for the neural network based link functions (FCN and CNN), we additionally include an architecture that matches the true underlying network for generation. Such an ``oracle'' specification naturally gives Adam an advantage in learning convolutional filters.

We evaluate convolutional filter estimation using the column space distance metric defined in (\ref{dist}). We repeatedly simulate the responses using four of the five folds of the data, and report the average and standard deviation of the estimation errors in Figure \ref{fig:adni_simu}. As shown in the figure, our method consistently achieves the smallest estimation errors across various link functions in both single and multiple filters scenarios, demonstrating its effectiveness even when the inputs are real MRI images with unknown distributions. Furthermore, these results suggest that Gaussian approximation can be sufficient for the MRI images in this application, providing practical support for the subsequent analysis with real MMSE responses.

\begin{figure}[!htbp]
	\centering
	\includegraphics[width=\linewidth]{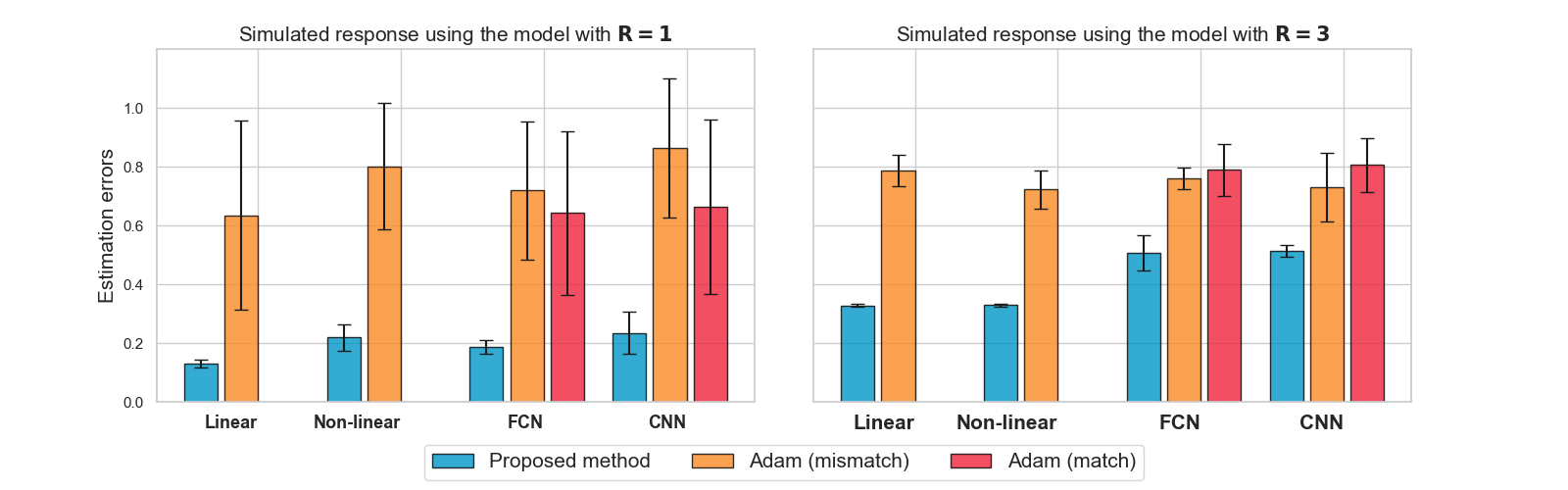}
	\caption{ADNI data analysis with simulated response. Filters estimation performance of the proposed plug-in estimator compared with Adam under different link functions.}
	\label{fig:adni_simu}
\end{figure}

\subsection{Real MMSE response and representation learning }\label{subsec:mmse}
We further use real MMSE scores as the response variable. Since the true latent convolutional filters are unknown, direct comparison of estimation accuracy is not feasible. Instead, we evaluate performance from a representation learning perspective by examining predictive performance across different representations and predictive models. That is to say, we consider models of the form $\hat Y_i = f(h(\bX^{\text{ori}}))$, where $h$ denotes a representation function and $f$ refers to the downstream predictive model.
Specifically, the following representations are included.
\begin{itemize}
	\item[(1)] Original features: $h(\bX^{\text{ori}}) = \bX^{\text{ori}}$.
	\item[(2)] Representations by our method: $h(\bX^{\text{ori}}) = (\bX^{\text{ori}}\star\bB_1, \ldots, \bX^{\text{ori}}\star\bB_R)$, where $\bB_r$ are directly estimated by our method.
	\item[(3)] Representations by Adam: $h(\bX^{\text{ori}}) = (\bX^{\text{ori}}\star\bB_1, \ldots, \bX^{\text{ori}}\star\bB_R)$, where $\bB_r$ are learned jointly with the full model by Adam.
	\item[(4)] Representations by TensorSIR: $h(\bX^{\text{ori}}) = \bGamma_1^\top\bX^{\text{ori}}\bGamma_2$, where $\bGamma_1$ and $\bGamma_2$ are estimated using tensor-valued sliced inverse regression method proposed by \cite{ding2015tensor}.
\end{itemize}

Accordingly, our procedure first estimates filters using the proposed method to obtain convolutional representations, and then trains the predictive models on these representations to perform prediction. For example, when the convolution filters $\bB_r$ are fixed to size $4\times 4$, each convolution operation reduces the input dimension from 2,304 ($=48\times 48$) to 144 ($=12\times 12$). Figure \ref{fig:vis_filter} illustrates the filters estimated by our method for $R=3$ along with the corresponding convolutional representations of a sample image. For the gradient-based approach, we adopt the same convolutional architecture as in our method, with the only difference being that the filters are learned jointly with the full model via Adam. As a result, the two types of representations have the same dimensions given the same $R$. TensorSIR is an extension of sliced inverse regression for tensor-valued inputs and is therefore naturally applicable to matrix-valued imaging data. In this setting, $\bGamma_1$ and $\bGamma_2$ are estimated to perform two-way dimension reduction along the row and column directions, respectively.

Based on the above representations, we consider both statistical models and neural networks as downstream predictive models, including the linear model, LASSO, a two-layer FCN, and a CNN. The architectures of the neural networks are kept the same as in the simulation studies. For each model type, we use a consistent architecture across different representations, although the number of parameters (or neurons) may vary depending on dimensions of representations. In particular, for the two types of low-dimensional representations with the same $R$, their predictive models are identical.

\begin{figure}[!htbp]
	\centering
	\includegraphics[width=0.7\linewidth]{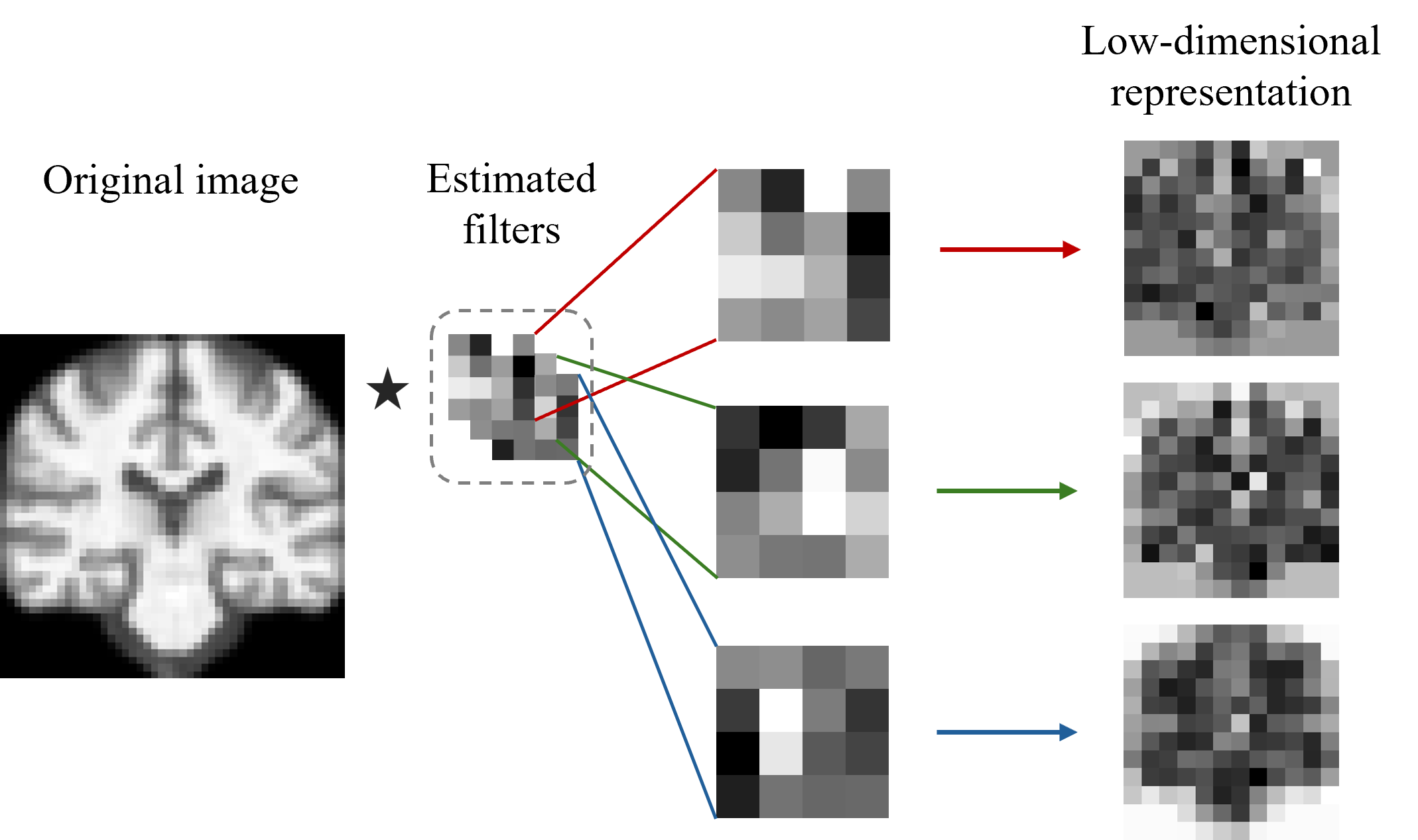}
	\caption{Visualization of the estimated filters and the low-dimensional representation for a sample image. The original $48\times 48$ image is reduced by the proposed method using $R=3$ estimated $4\times 4$ filters, producing three corresponding feature maps of size $12\times 12$.}
	\label{fig:vis_filter}
\end{figure}

We evaluate performance using 10-fold cross-validation, with 90\% of the data used for training and 10\% for testing in each fold. For methods involving convolutional representations, we select $R\in\{1, 2, 3\}$ based on a 20\% validation split from the training portion of each fold. The model is then retrained on the full training portion using the optimal $R$ to make final predictions for the test portion. In implementation, inputs are vectorized for the first three models, while they are retained in matrix form for the CNN. Additional implementation details are provided in the Supplementary Material. Figure \ref{fig:real} presents the average prediction errors, measured by root mean square error (RMSE). The results show that, across a range of predictive models, the convolutional representations learned by our method consistently improve prediction accuracy, demonstrating the effectiveness and broad applicability of the proposed approach.

\begin{figure}
	\centering
	\includegraphics[width=\linewidth]{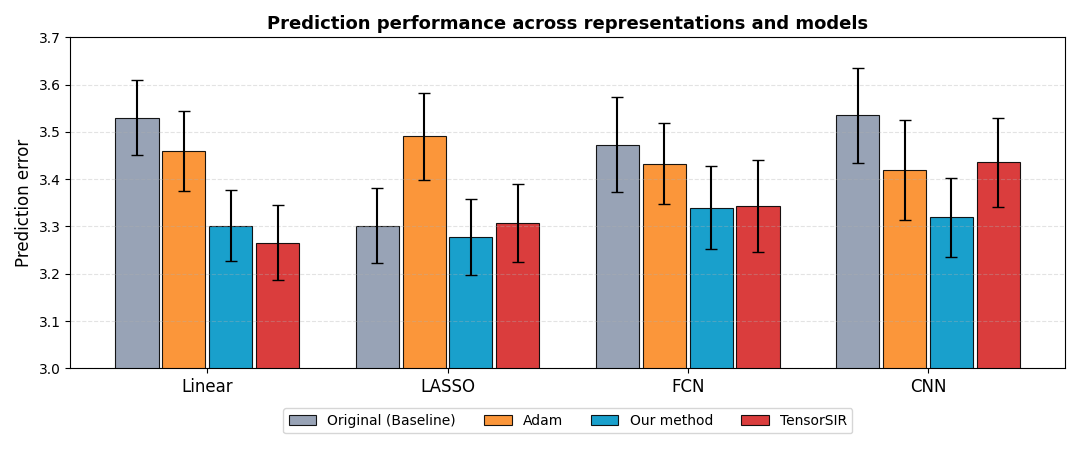}
	\caption{Prediction performance of different representations across various predictive models.}
	\label{fig:real}
\end{figure}

Moreover, when LASSO is used as the predictive model, the estimated sparse coefficients can be reshaped back to matrices with the original dimensions for interpretability. For the original features and TensorSIR representations, the coefficient matrix in each fold naturally serves as a saliency map. While for convolutional representations with $R$ filters, the coefficients in each fold can be reshaped back to $R$ corresponding coefficient matrices as salience maps. Each element in these matrices reflects the importance of convolutional representations, which correspond to local blocks in the original image rather than individual pixels. We aggregate these $R$ maps into a single saliency map by taking the element-wise maximum, treating a location as important if it appears in any of the feature maps. For each type of representation, the final saliency map is obtained by intersecting the saliency maps across the 10 folds of cross-validation, with nonzero entries averaged over the folds. For TensorSIR, where the intersection is empty, we instead use a $k$-occurrence masks with $k=6$ to produce an appropriately sparse map.

We visualize these saliency maps in Figure \ref{fig:salicene_maps} by overlaying them on a sample image to highlight the important regions identified by different methods. Note that the saliency map from the original features has size $48\times 48$ whereas those from other representations have size $12\times 12$, so we zoom it back to the original dimension. The results show that, based on representations learned by our method, LASSO identifies regions primarily around the hippocampus. Its progressive atrophy is closely associated with early memory decline and is widely recognized as a key biomarker of Alzheimer’s disease \citep{rao2022hippocampus, disouky2026human}. While saliency maps from other representations also roughly cover these regions, those based on the original features (left subplot) appear overly dispersed, and those based on representations by Adam and TensorSIR (right two subplots) are comparatively weak or highlight less relevant regions, including spurious areas near the image boundaries. These findings suggest that our method integrates effectively with interpretable statistical models and enhances interpretability, making it particularly suitable for medical imaging analysis.

\begin{figure}[!htbp]
	\centering
	\includegraphics[width=\linewidth]{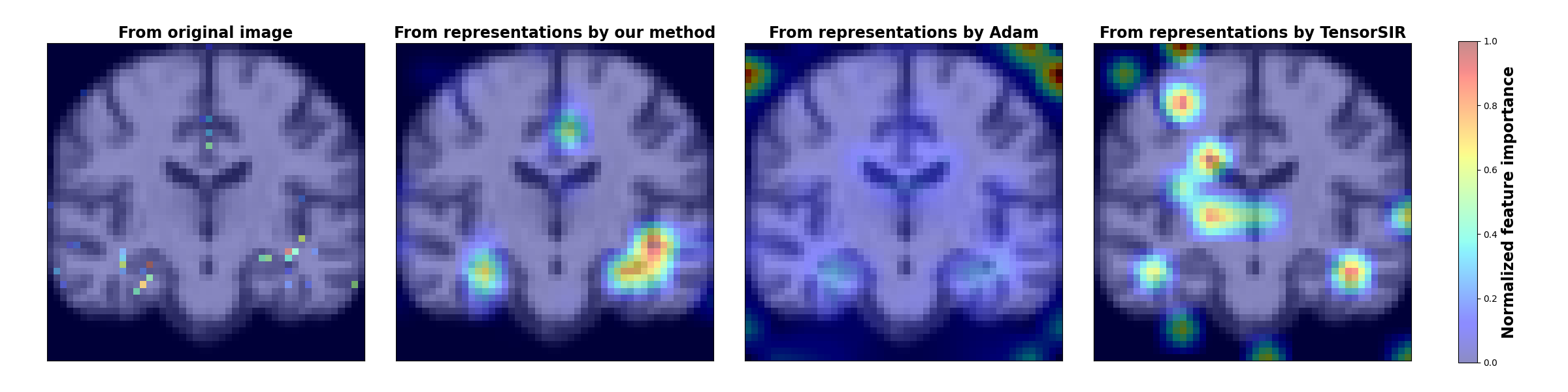}
	\caption{Saliency maps from different representations when LASSO as predictive models.}
	\label{fig:salicene_maps}
\end{figure}

\section{Simulation studies}\label{sec:simu}
We conduct comprehensive simulation studies to evaluate the performance of the proposed estimator. Both the single filter and multiple filters settings described as (i) and (ii) in Section \ref{subsec:simu_y} are considered, with the same convolutional filters and noise specifications. Throughout the simulations, the input image size is fixed at $28\times 28$ and the filter size at $4\times 4$, yielding reshaped images $\bX_i=\mathcal{R}_{(4,4)}(\bX^{\text{ori}})$ of size $49\times 16$. The sample size $n$ ranges from 500 to 2500. We vary the distributions of input $\bX_i^{\text{ori}}$ and link function $f$. First, we examine the following distributional settings of $\bX^{\text{ori}}$:
\begin{itemize}
	\item[(1)] with i.i.d. Gaussian entries, $\bX^{\text{ori}}_{jk}\sim \mathcal{N}(\mu, \sigma^2)$, $\mu = 0$ and $\sigma = 1$.
	\item[(2)] with i.i.d. Student's $t$ entries: $\bX^{\text{ori}}_{jk}\sim t(\nu)$, $\nu = 5$.
	\item[(3)] with i.i.d. Gamma entries: $\bX^{\text{ori}}_{jk}\sim \Gamma(\alpha, \beta)$ with $\alpha=5$ and $\beta = 1$.
	\item[(4)] with correlated Gaussian entries. Specifically, we generate $\bX^{\text{ori}}$ by $\text{vec}(\bX_i^{\text{ori}}) \sim \mathcal{N}(\bmu, \bSigma)$. The mean vector $\bmu$ has uniformly random integer entries ranging from $[-5,5]$, while the covariance matrix $\bSigma$ has entries $\bSigma_{j,k} = \rho^{|j-k|}$, with $\rho=0.2$.
\end{itemize}

Furthermore, we consider the following types of link function $f$:
\begin{itemize}
	\item[(\rom{1})] Linear: $f(\bZ) = \langle\bZ, \bC\rangle$, where $\bC$ has i.i.d. Gaussian entries $\bC_{jk}\sim \mathcal{N}(0, 1)$. 
	\item[(\rom{2})] Nonlinear: $f(\bZ) = \langle\bZ+ 3\sin(\bZ), \bC\rangle$, where $\bC$ is defined the same as in (\rom{1}). Here the function $\sin(\bZ)$ is applied element-wisely.
	\item[(\rom{3})] Fully connected network (FCN): $f(\bZ)$ is a two-layer fully connected network. Thus the original output $Y_i$ is generated by a standard CNN with one non-overlapping convolutional layer and two fully connected layers.
	\item[(\rom{4})] Convolutional neural network (CNN): $f(\bZ)$ is a convolutional neural network. We defer to the Supplementary Material for a detailed description of the structure of the CNN and additional details on this simulation study.
\end{itemize}

Next, we present the results for the multiple-filter setting. The single-filter setting, being a special case, is deferred to the Supplementary Material.

\subsection{Comparison with Adam}\label{subsec:adam}
In this subsection, we focus on the original estimator under the assumption that the input distribution is known. We acknowledge that the comparison with Adam has an inherent asymmetry: the proposed estimator uses distributional information, whereas Adam is a general optimization algorithm and does not explicitly use such information. Nevertheless, this comparison is meaningful as it illustrates how much can be gained in filter estimation when distributional structure is available and used statistically.
The comparison with Adam is conducted following the same setup as in Section \ref{subsec:simu_y}. Specifically, the link functions are approximated consistently using two-layer fully connected networks, and when fitting FCN and CNN link functions, both matched and mismatched architectures are incorporated.
Figure \ref{fig:mim} shows the convolution filter estimation performance of our method and Adam in the multiple-filter setting, based on 10 independent repetitions. In addition, computational efficiency comparisons are provided in the Supplementary Material, which clearly demonstrate the substantial advantage of the proposed method.

\begin{figure}[!htbp]
	\centering
	\includegraphics[width=\textwidth]{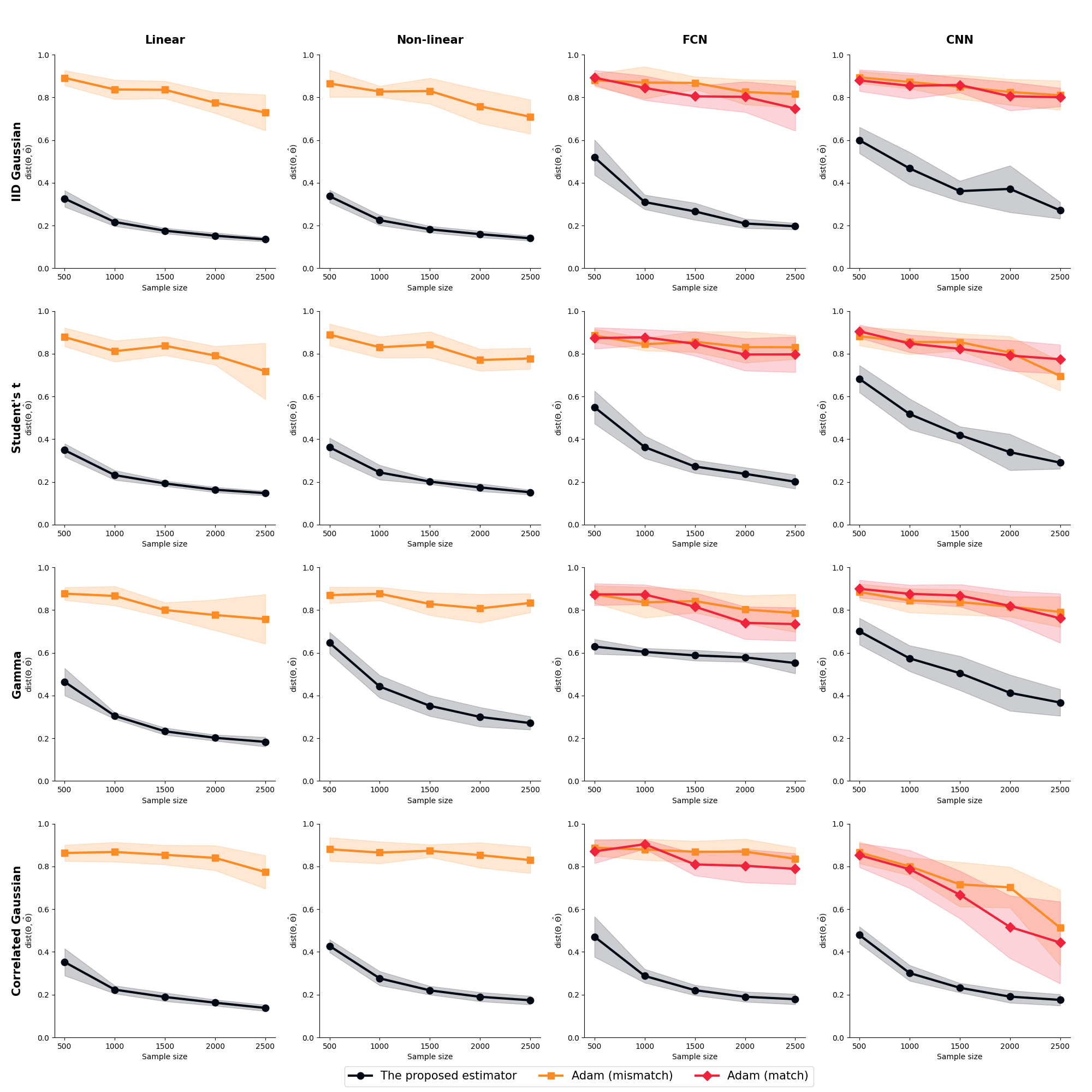}
	\caption{Comparison of multiple filters estimation performance under different link functions and input distributions.}
	\label{fig:mim}
\end{figure}

Figure \ref{fig:mim} clearly shows that our approach consistently performs well and dominates Adam across various link functions and input distributions. In particular,  even for the matched case where the neural network structures are correctly specified in cases (\rom{3}) and (\rom{4}), our approach still exhibits significant advantages, especially when the sample size $n$ is relatively small.
As the sample size increases, the performance of Adam improves, leading to a decrease in filter estimation error. As mentioned, neural network training typically requires a large sample size to achieve desired performance. In applications such as medical imaging, where data are often limited, our approach can be particularly advantageous due to its much weaker dependence on sample size.

\subsection{Effect of truncation}\label{sec:simtruncation}
As discussed earlier, when the data explicit heavy-tailed behavior, the original estimator may not be optimal due to a lack of concentration on the expectation of $YS(\bX)$. In this subsection, we compare the performance of the original and the truncated version of the estimator under different distribution assumptions. For the truncated estimator, we apply the non-decreasing function $\phi(\cdot)$ in (\ref{phi}). We consider the input with heavier-tailed distributions, such as the Student's $t$ and Gamma distributions. However, it is worth noting that while both of them can display heavy-tailed behavior, their scores differ. The score of a Student's $t$ distribution is bounded and inherently sub-Gaussian, whereas the score of a Gamma distribution remains heavy-tailed.

\begin{figure}[!htbp]
	\centering
	\includegraphics[width=\textwidth]{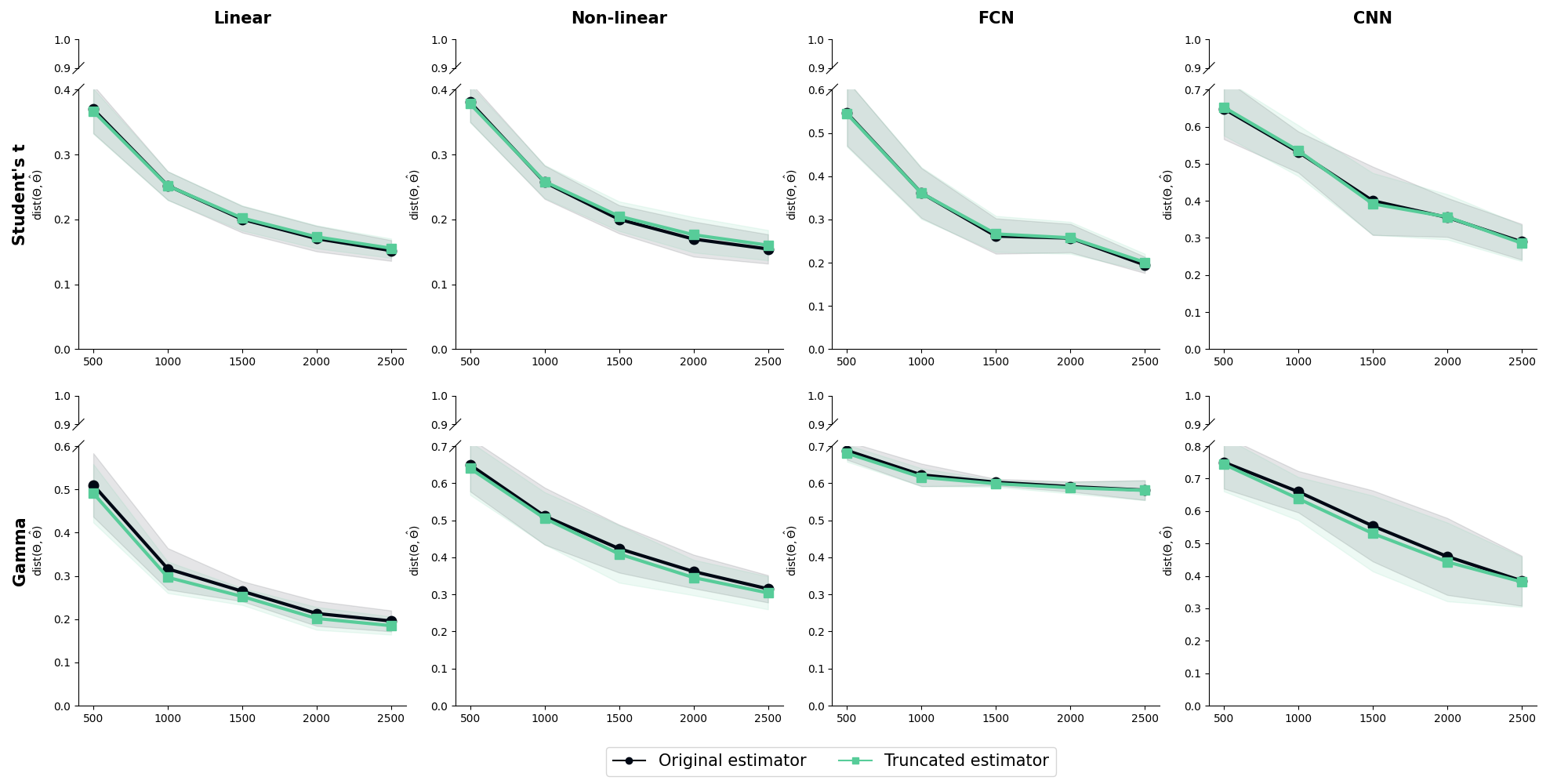}
	\caption{The original estimator vs the truncated estimator under the Student's $t$ and Gamma distribution with different link functions.}
	\label{fig:truncation}
\end{figure}

Figure \ref{fig:truncation} presents the estimation error of both the original and truncated estimators. It shows that the effect of truncation is relatively limited in the case of the Student's $t$ distribution, as the original and truncated estimators achieve similar estimation errors. In contrast, the truncated estimator generally achieves smaller estimation errors for the Gamma distribution case. This suggests the effect of truncation when handling heavy-tailed scores.

\subsection{Performance of plug-in estimator}\label{sec:exp:emp}
In this subsection, we consider the scenario where the input distribution is unknown. To illustrate the performance of the plug-in estimator, we focus on Gaussian distributions with unknown means and covariance matrices. Therefore, we revisit the case (4), where the input has correlated Gaussian entries, i.e., $\bX^{\text{ori}}\sim \mathcal{N}(\mu,\bSigma)$, with covariances $\bSigma_{i, j} = \rho^{|i-j|}$. Two values of $\rho$ are examined, namely $\rho=0.5$ and $\rho=0.8$, representing different levels of correlation. We evaluate the performance of the plug-in estimator introduced in Section \ref{sec:emp} by comparing it with the original estimator where the mean and covariance matrix are correctly specified. Figure \ref{fig:empirical} shows the estimation errors of both estimators for different $\rho$ values.

\begin{figure}[!htbp]
	\centering
	\includegraphics[width=\textwidth]{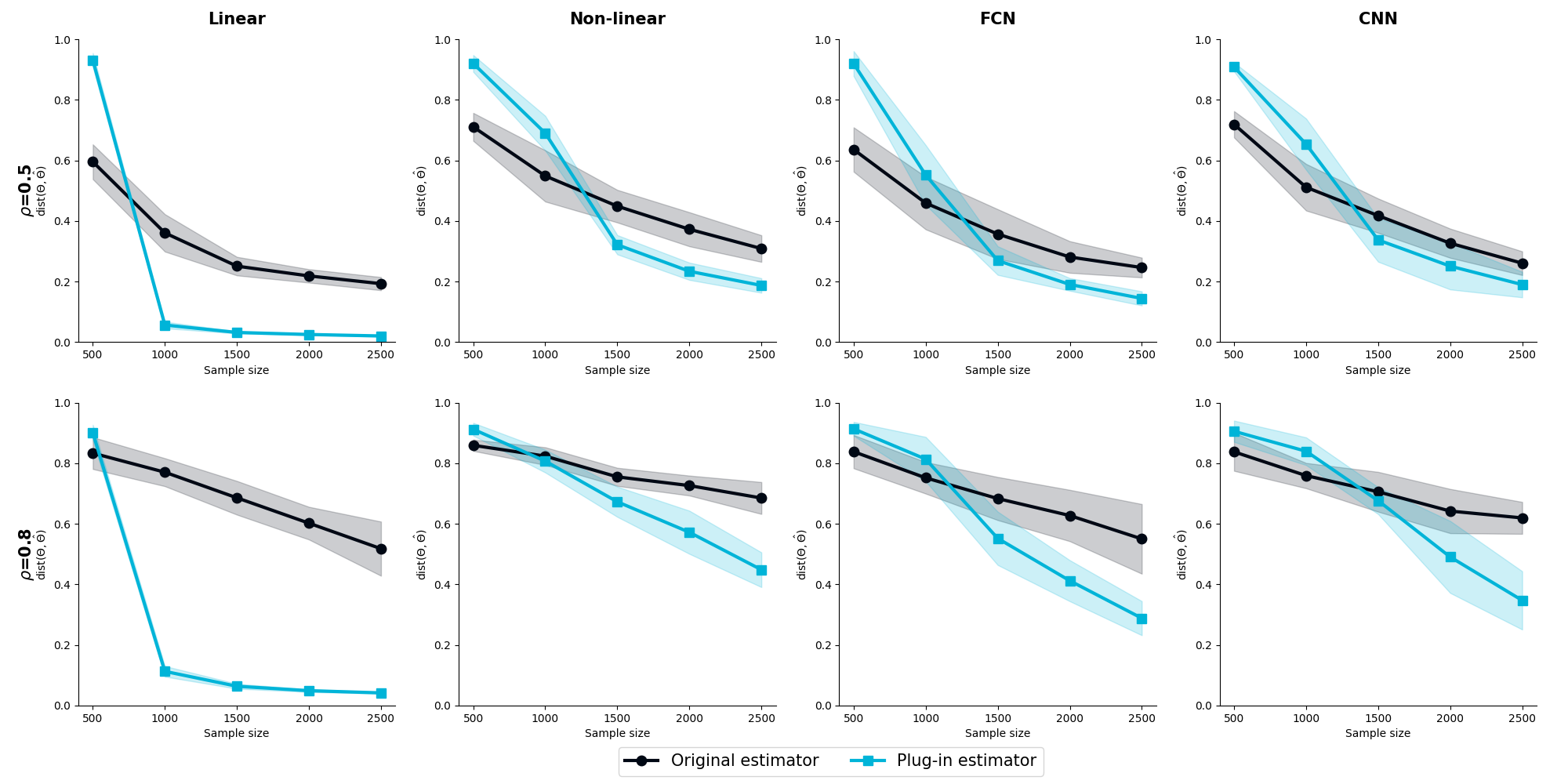}
	\caption{The original truth-based estimator vs the plug-in estimator under different covariance matrices. }
	\label{fig:empirical}
\end{figure}

By Figure \ref{fig:empirical}, we are surprised to find that the plug-in estimator performs no worse, in fact even better,  than the original estimator, particularly when the sample size $n$ increases. 
Indeed, as discussed in Section \ref{sec:emp}, in the Gaussian scenario, the plug-in estimator can be advantageous by offering reduced bias in linear models. This simulation further demonstrates that these advantages could extend beyond linear models, highlighting the adaptability of our approach when dealing with unknown parameters.

\section{Discussions}\label{sec:conc}
Motivated by the problem of learning convolutional filters in CNNs, we propose a statistical method to estimate filters in the first convolutional layer without training the full neural network, thereby providing an efficient training-free approach for learning convolutional representations. Promising future directions include extensions to overlapping convolutions and, more broadly, to deeper architectures.

Although our study focuses on non-overlapping convolutions, it is also possible to extend our approach to the overlapping case. Let $*$ denote a general convolution operator with row and column strides $s_1$ and $s_2$, respectively. Let $\bB_1,...,\bB_R\in\RR^{d_1 \times d_2}$ be the filters in the first layer of a CNN. Assume that $s_1,s_2$ are factors of $P_1-d_1, P_2-d_2$ respectively.
For any matrix $\bM\in\RR^{P_1 \times P_2}$ and a filter $\bB\in\RR^{d_1 \times d_2}$, the overlapping convolution is defined as 
\begin{align*}
	\bM\ast\bB \in \mathbb{R}^{\tilde{p}_1 \times \tilde{p}_2}, \ \ (\bM\ast\bB)_{j,k} = \left\langle\tilde{\bM}_{j, k}^{d_1,d_2,s_1,s_2}, \bB \right\rangle,
\end{align*}
where $\tilde{p}_1=(P_1-d_1)/s_1+1$, $\tilde{p}_2=(P_2-d_2)/s_2+1$, and $\tilde{\bM}_{j,k}^{d_1,d_2,s_1,s_2}$ denotes the $(j,k)$-th overlapping patch of size $d_1 \times d_2$ for $j\in[\tilde{p}_1]$ and $k\in[\tilde{p}_2]$. Similar to \eqref{Rmat}, define the reshaping operator $\tilde{\mathcal{R}}_{(d_1,d_2,s_1,s_2)}: \RR^{P_1 \times P_2} \rightarrow \RR^{(\tilde{p}_1\tilde{p}_2) \times (d_1d_2)}$ as
\begin{align*}
	\tilde{\mathcal{R}}_{(d_1,d_2,s_1,s_2)}(\bM) = \left[\text{vec}\left(\tilde{\bM}_{1,1}^{d_1,d_2,s_1,s_2}\right), \ldots, \text{vec}\left(\tilde{\bM}_{p_1,p_2}^{d_1,d_2,s_1,s_2}\right)\right]^\T.
\end{align*}
Denote $\tilde{\bX}=\tilde{\mathcal{R}}_{(d_1,d_2,s_1,s_2)}(\bX^{\text{ori}})$ and $\bb_r=\text{vec}(\bB_r)\in\RR^{d_1d_2}$, the output can be written as 
\begin{equation*}
	G(\bX^{\text{ori}}) = f(\bX^{\text{ori}}\ast\bB_1, \ldots, \bX^{\text{ori}}\ast\bB_R)=f(\tilde{\bX}\bb_1,...,\tilde{\bX}\bb_R).
\end{equation*}
Assuming the model $\EE(Y|\tilde{\bX}) = f(\tilde{\bX}\bb_1,...,\tilde{\bX}\bb_R)$, Stein's identity can still be applied to estimate the column space of $(\bb_1,\cdots, \bb_R)$, provided that the score function $S(\tilde{\bX})$ is carefully estimated, with particular attention to the dependence induced by overlapping patches.

As for further extension to deep architectures, the intuition is that features produced by earlier layers can be treated as inputs to subsequent layers, to which the proposed method can be applied again to estimate the corresponding filters. In this sense, the method can be extended to deeper architectures in a sequential manner. However, this also points to the main challenges. Accurately estimating the score function of intermediate features is nontrivial, as they have undergone nonlinear transformations. Moreover, estimation errors may accumulate across layers, making it difficult to control the overall error, especially in deep architectures. Addressing these challenges remains an important direction for future research. Nevertheless, extending the method to multiple convolutional layers is promising, as it may provide a viable statistical alternative to gradient-based optimization methods that have long been used in modern neural network training.

\bibliographystyle{apalike}
\bibliography{reference}


\clearpage

\specialsection*{Supplementary Material}

\setcounter{equation}{0}
\renewcommand{\theequation}{S\arabic{equation}}
\setcounter{figure}{0}
\renewcommand{\thefigure}{S\arabic{figure}}
\setcounter{table}{0}
\renewcommand{\thetable}{S\arabic{table}}
\setcounter{section}{0}
\renewcommand{\thesection}{S\arabic{section}}
\setcounter{theorem}{0}

In the supplementary material, we provide additional experimental details, as well as theoretical proofs of the theorems and lemmas presented in Sections~\ref{sec:model}, Section~\ref{sec:single_filter} and Section~\ref{sec:multiple_filter}.


\section{Experiment details}\label{app:exp}

\subsection{Additional results}
In Section \ref{sec:simu}, we also consider single filter scenario $Y_i = f\left(\bX_i^{\text{ori}} \star \bB\right) + \varepsilon_i$, which is the counterpart of the multiple filters scenario.
The results on single filter scenario including comparison to gradient descent, effect of truncation and performance of plug-in estimator, are shown in Figures \ref{fig:sim}-\ref{supp:fig:empirical}, respectively.
\begin{figure}[!htbp]
    \centering
    \includegraphics[width=\textwidth]{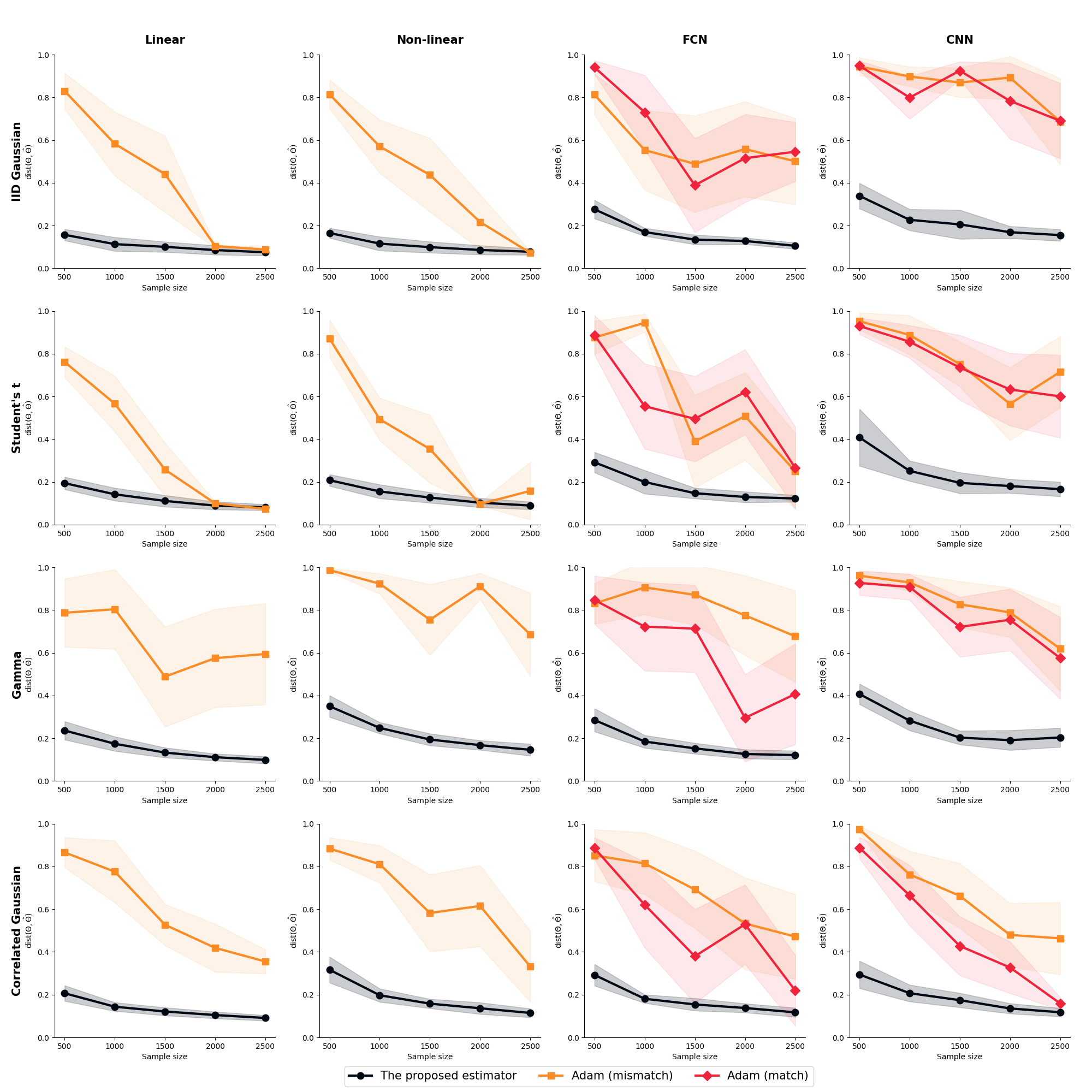}
    \caption{Comparison of single filter estimation performance under different link functions and input distributions.}
    \label{fig:sim}
\end{figure}

\begin{figure}[!htbp]
    \centering
    \includegraphics[width=\textwidth]{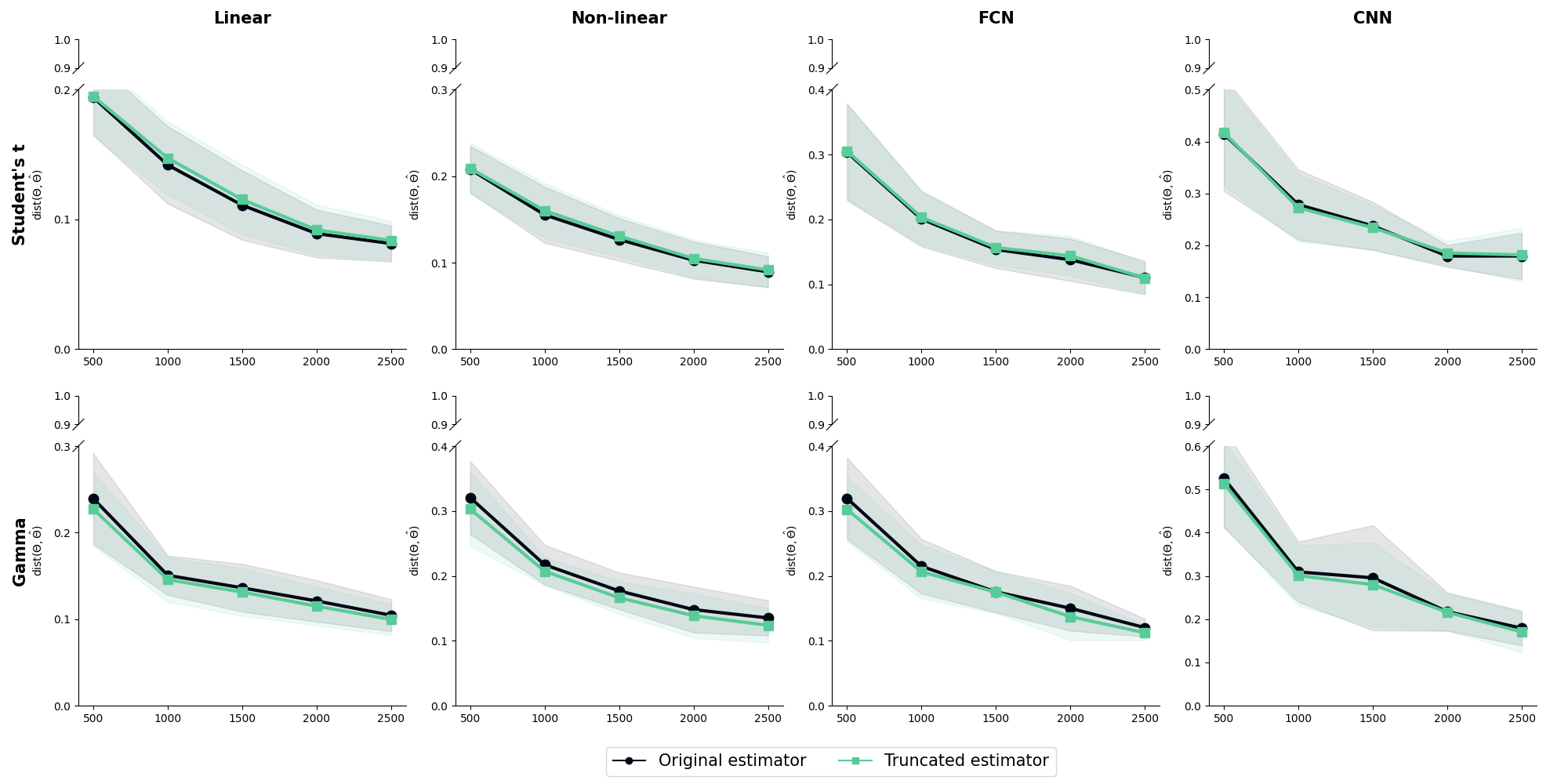}
    \caption{The original estimator vs the truncated estimator under the Student's $t$ and Gamma distribution with different link functions.}
    \label{supp:fig:truncation}
\end{figure}

\begin{figure}[!htbp]
    \centering
    \includegraphics[width=\textwidth]{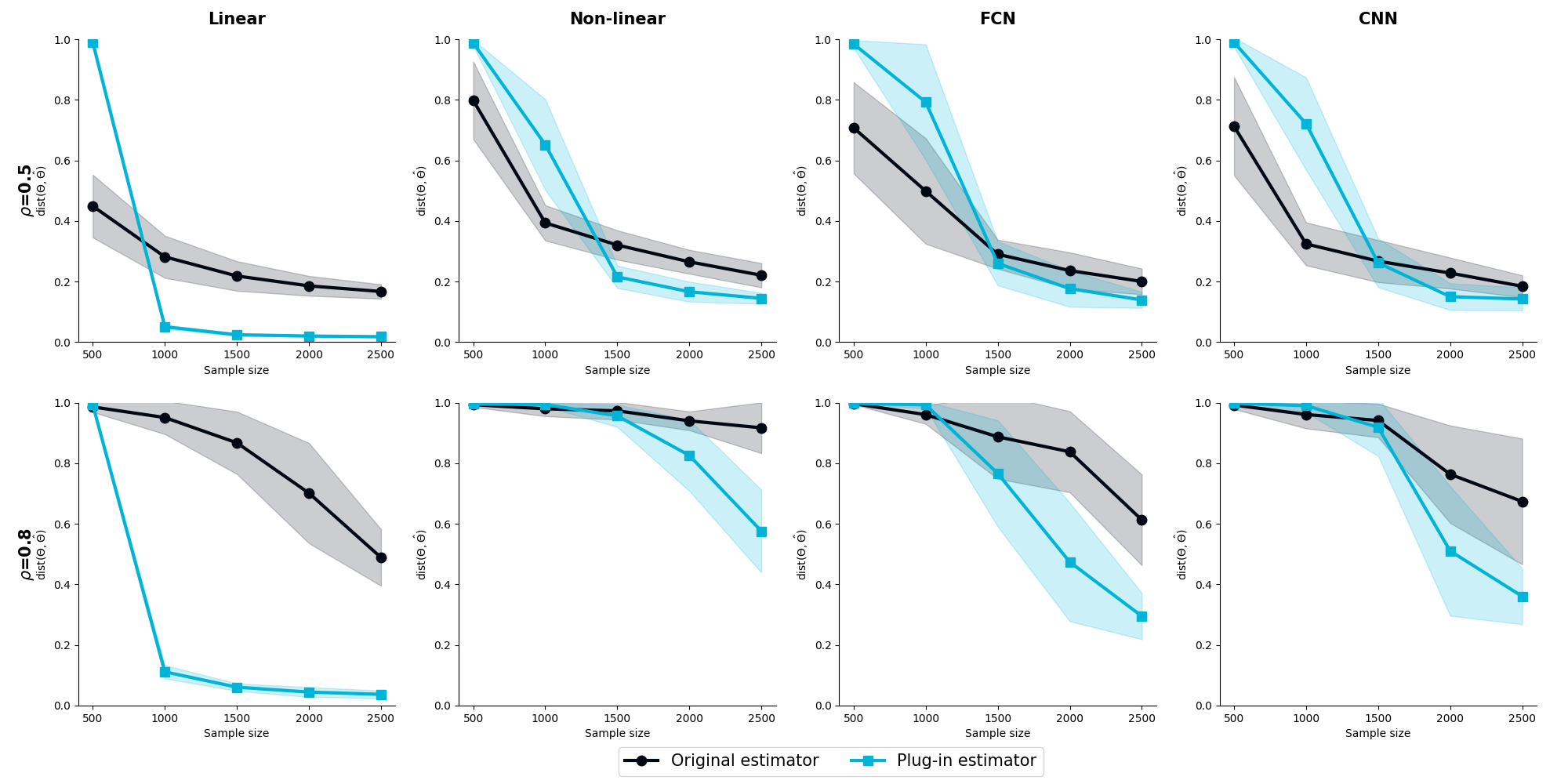}
    \caption{The original truth-based estimator vs the plug-in estimator under different covariance matrices. }
    \label{supp:fig:empirical}
\end{figure}

In addition, to compare computational efficiency, we report the time cost under the first setting of the single filter scenario in Table \ref{supp:tab:time} for illustration. The results clearly demonstrate the substantial computational advantage of the proposed method.
\begin{table}[!htbp]
\caption{Average computational time (in seconds) under the first case. The results for other cases are highly similar.}
\begin{center}
\begin{tabular}{cccccc}
\toprule
     & 500           & 1000          & 1500           & 2000           & 2500           \\ \midrule
Proposed method & 0.023 & 0.045 & 0.063 & 0.090 & 0.107  \\ \hline
Adam & 3.236 & 7.161 & 13.564 & 18.060 & 22.537 \\ \bottomrule
\end{tabular}
\end{center}
\label{supp:tab:time}
\end{table}

\subsection{Implementation}\label{app:simu}
We present the specific neural network architectures used in the simulations for both data generation and model fitting. Since the matching architectures used for fitting mirror those in the generation stage, we omit their repetition and display only the mismatched counterparts.

The FCN begins with a ReLU activation applied to the convolutional features, followed by two fully connected layers, where the first layer reduces the dimension of the flattened convolutional features by half. For an input of size $28 \times 28$, a single $4\times 4$ convolutional filter produces an output of size $7\times 7$, leading to $49$ nodes in the subsequent fully connected layer. Accordingly, the two fully connected layers have $49\to 24$ nodes. For the mismatched network, the only change is that the first fully connected layer preserves the original dimension instead of halving it, so the two fully connected layers have $49 \to 49$ nodes.

The CNN consists of a typical convolution layer, batch normalization layer, ReLU function and max-pooling layer, followed by two fully-connected layers of which the first one reduces the dimension by half. Similarly, the difference of mis-matching network is that the first fully-connected layer keeps the dimension rather than halving. The structures of networks used are shown in Figure \ref{fig:networks}.

\begin{figure}[!htbp]
    \centering
    \includegraphics[width=\linewidth]{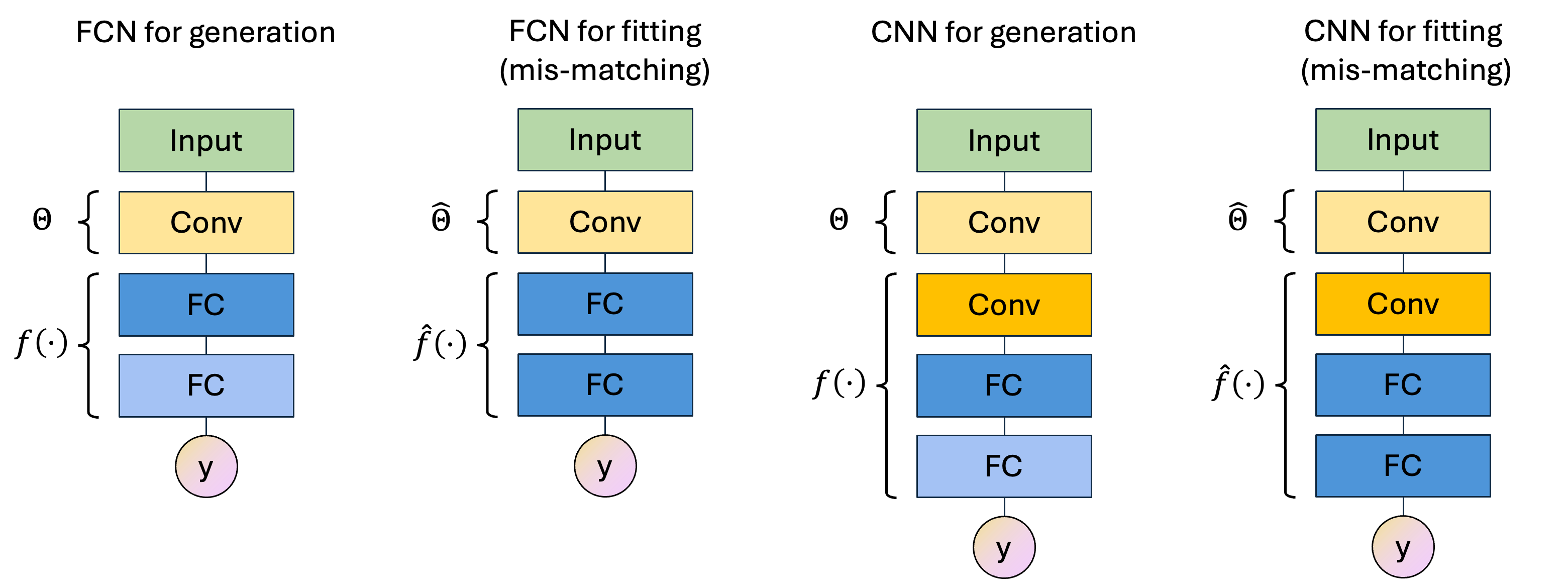}
    \caption{Structures of networks used.}
    \label{fig:networks}
\end{figure}

For real applications, we use the same models for generation as in the simulation study to produce simulated responses, and we examine both matched and mismatched architectures for FCN and CNN link functions. For the analysis of real responses, both the FCN and CNN models employ three fully connected layers with identical widths to enhance model capacity. The training procedure follows exactly the same setup as in the simulation study.
In addition, the $\ell_1$ regularization parameter in LASSO is tuned by 5-fold cross-validation over the range $[0.01, 0.1]$ with a step size of $0.01$.

All neural network models were implemented in PyTorch. We used the Adam optimizer with its default settings and a batch size of 128. Each model was trained for up to 500 epochs with early stopping based on the training loss. In all experiments, our methods were executed on Intel Xeon Gold 5218R CPUs, while the Adam optimizer-based models were trained on NVIDIA GeForce RTX 3090 GPUs.


\subsection{Signal strength detection}
In practice, Assumptions \ref{assumption:multiple1} and \ref{assumption:sub3} can be assessed by examining the empirical singular value spectrum of the matrix $\frac{1}{n}\sum_{i=1}^{n}Y_iS(\bX_i)$. Specifically, we can compute the singular values $\hat \sigma_1 \ge \hat \sigma_2 \ge \cdots$ and inspect their decay pattern. The assumptions are supported when the leading $R$ singular values are clearly separated from the remaining spectrum, i.e., there is a noticeable spectral gap between $\hat \sigma_R$ and $\hat \sigma_{R+1}$, indicating a stable low-rank signal structure. In contrast, a rapid decay without a clear gap suggests that the effective signal strength is weak, and the column space may not be reliably identifiable. As Assumption \ref{assumption:sub3} corresponds to the special case $R=1$, the same diagnosis reduces to checking whether the leading singular value is significantly larger than the remaining spectrum, consistent with a rank-one signal structure. We illustrate the diagnosis under the simulation settings in Section \ref{sec:exp:emp} and present in Figure \ref{fig:screeplots} the singular values of matrices $\hat \bA$ and its plug-in counterpart $\check \bA$. The results show that the plug-in version closely matches the spectrum of $\hat \bA$ and clear spectral gaps appear at the $R$-th singular value for both the single-filter case ($R=1$) and the multiple-filter case ($R=3$). Moreover, the presence of such spectral gaps can also serve as a practical tool for selecting $R$ when the number of filters is treated as unknown or not pre-specified.

\begin{figure}[!htbp]
    \centering
    \includegraphics[width=\linewidth]{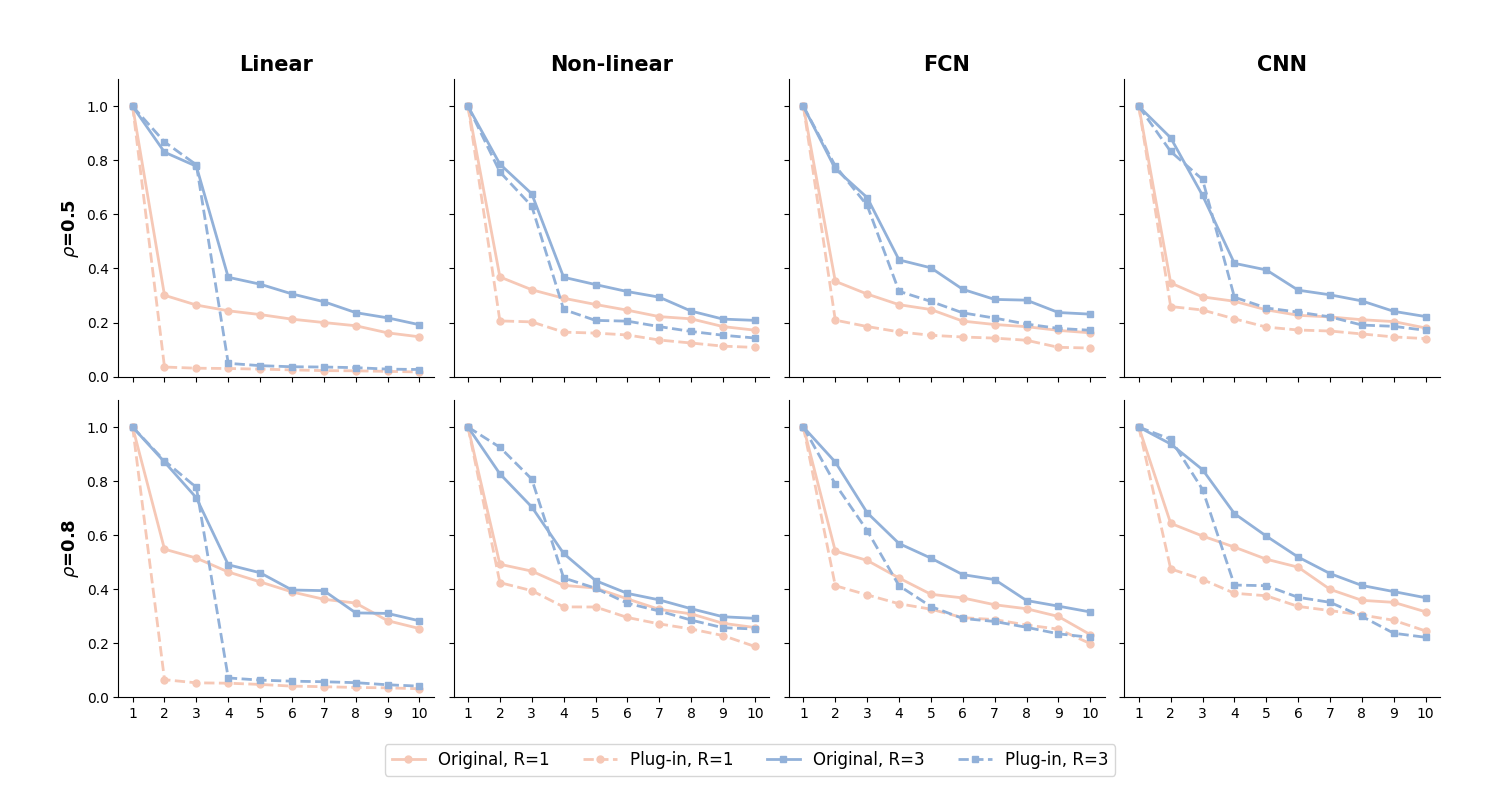}
    \caption{Singular values of $\hat \bA$ and $\check \bA$ under the simulation setting in Section \ref{sec:exp:emp} with sample size $n=2500$, including both cases with $R=1$ and $R=3$. The singular values are normalized by the largest singular value, and the first 10 values are displayed.}
    \label{fig:screeplots}
\end{figure}

\section{Proofs}\label{app:proof}

In this section, we provide the proofs of the theorems and lemmas in Section~\ref{sec:model}, Section~\ref{sec:single_filter} and Section~\ref{sec:multiple_filter}.

\subsection{Proof of Lemma~\ref{lemma:CNNexpectation}}
In order to prove Lemma~\ref{lemma:CNNexpectation}, we first present the following lemma as an extension of Non-Gaussian Stein's Lemma \citep{stein2004use}. This version applies Stein's Lemma to non-Gaussian random matrices.

\begin{lemma}\label{lemma:stein_lemma}
Let $g:\RR^{p\times d}\rightarrow\RR$ be a continuously differentiable function and $\bX\in\RR^{p\times d}$ be a random matrix with density $P$ which is also continuously differentiable. Under the assumption that $\EE[g(\bX) \cdot S_{ij}(\bX)]$ and $\EE[\nabla_{X_{ij}} g(\bX)]$ are well-defined, we have 
\begin{align*}
\EE[g(\bX) \cdot S_{ij}(\bX)]=\EE[\nabla_{X_{ij}} g(\bX)].
\end{align*}
\end{lemma}

\begin{proof}[\textbf{Proof of Lemma~\ref{lemma:stein_lemma}}]
We know $S_{ij}(\bX)=-\nabla_{X_{ij}}P(\bX)/P(\bX)$. Thus, we have that
\begin{align*}
\EE[g(\bX) \cdot S_{ij}(\bX)] &= \int_{\RR^{p \times d}} g(x)S_{ij}(x)P(x)dx = -\int_{\RR^{p \times d}} g(x)\nabla_{x_{ij}}P(x)dx.
\end{align*}
By the integration by parts formula, we get
\begin{align*}
\int_{\RR^{p \times d}} g(x)\nabla_{x_{ij}}P(x)dx &= \int_{\RR^{p \times d}} \nabla_{x_{ij}}\big[g(x)P(x)\big]dx - \int_{\RR^{p \times d}} \nabla_{x_{ij}}g(x)P(x)dx\\
&= - \int_{\RR^{p \times d}} \nabla_{x_{ij}}g(x)P(x)dx\\
&= - \EE[\nabla_{X_{ij}} g(\bX)].
\end{align*}
Combining two equations, we can obtain that
\begin{align*}
\EE[g(\bX) \cdot S_{ij}(\bX)] = \EE[\nabla_{X_{ij}} g(\bX)].
\end{align*}
\end{proof}

Next, we are ready to prove Lemma~\ref{lemma:CNNexpectation}.

\begin{proof}[\textbf{Proof of Lemma~\ref{lemma:CNNexpectation}}]
    By Lemma~\ref{lemma:stein_lemma}, assuming that the link function $f: \bz \mapsto f(\bz)$ and the score function $S(\bX)$ satisfy the conditions stated in Lemma~\ref{lemma:stein_lemma}, we can get 
\begin{align*}
\EE[Y \cdot S_{ij}(\bX)] &= \EE[f(\bX\bTheta) \cdot S_{ij}(\bX)]=\EE\{\nabla_{X_{ij}} [f(\bX\bTheta)]\} = \EE\left[\frac{\partial f}{\partial \bz_{i\cdot}}(\bX\bTheta)\right] \cdot \bTheta_{j\cdot}^{\top},
\end{align*}
where $\bz_{i\cdot}$ is the $i$-th row of $\bz$ and $\bTheta_{j\cdot}$ is the $j$-th row of $\bTheta$. Let $S_{\cdot j}(\bX)$ be the $j$-th column of $S(\bX)$. We can get that for all $j\in[d]$,
\begin{align*}
\EE[Y \cdot S_{\cdot j}(\bX)] = \EE\left[\nabla_{\bz} f (\bX\bTheta)\right] \cdot \bTheta_{j\cdot}^{\top}.
\end{align*}
Then, combining all the vectors into a matrix, we obtain that
\begin{align*}
\EE[Y \cdot S(\bX)] = \EE\left[\nabla_{\bz} f (\bX\bTheta)\right] \cdot \bTheta^{\top}. 
\end{align*}
\end{proof}

\subsection{Proof of Theorem~\ref{thm:subbeta}}
We first show the convergence rate of $\frac{1}{n}\sum_{i=1}^{n}Y_iS(\bX_i)$ by the following lemma.

\begin{lemma}\label{thm:subA}
Suppose the model $\EE(Y_i)=f(\bX_i\bTheta)$ for $i\in[n]$. Under Assumption~\ref{assumption:sub1}, Assumption~\ref{assumption:sub2} and Assumption~\ref{assumption:sub3}, with probability at least $1-\delta$, we have 
\begin{align*}
\left\|\frac{1}{n}\sum_{i=1}^{n}Y_iS(\bX_i)-\EE[YS(\bX)]\right\|_{F} = \mathcal{O}\left(\sqrt{\frac{pd}{n}\cdot\log\frac{pd}{\delta}}\right).
\end{align*}
\end{lemma}

\begin{proof}[\textbf{Proof of Lemma~\ref{thm:subA}}]
Since $S_{ij}(\bX)$ is a sub-Gaussian random variable, we can get that for any $i\in[p],j\in[d]$, there exists a positive constant $L$ such that $(\EE|S_{ij}(\bX)|^k)^\frac{1}{k}\leq L\sqrt{k}$ for $k\geq 1$. We first prove that $YS_{ij}(\bX)-\EE[YS_{ij}(\bX)]$ is a sub-exponential random variable. To simplify the notation, in the following proof, we denote $f=f(\bX\bTheta)$ and $s=S_{ij}(\bX)$. Then, for $k\geq 1$, we have that
\begin{align*}
\big[\EE|Ys-\EE(Ys)|^k\big]^\frac{1}{k} &\overset{(i)}{=} \big[\EE|(f+\varepsilon)s-\EE(fs)|^k\big]^\frac{1}{k}\\
&\overset{(ii)}{\leq} (\EE|fs|^k)^\frac{1}{k}+(\EE|\varepsilon s|^k)^\frac{1}{k}+|\EE(fs)|\\
&\overset{(iii)}{\leq} (\EE f^{2k} \cdot \EE s^{2k})^\frac{1}{2k}+(\EE |\varepsilon|^{k} \cdot \EE |s|^{k})^\frac{1}{k}+(\EE f^2 \cdot \EE s^2)^\frac{1}{2}\\
&\overset{(iv)}{\leq} T\sqrt{2k} \cdot L\sqrt{2k} + \sigma_\varepsilon\sqrt{k} \cdot L\sqrt{k} + T\sqrt{2} \cdot L\sqrt{2}\\
&\overset{(v)}{\leq} 4TLk + \sigma_\varepsilon Lk\\
&= (4T+\sigma_\varepsilon)Lk,
\end{align*}
where $(i)$ is by $\EE(\varepsilon s)=\EE\varepsilon\cdot\EE s=0$, $(ii)$ is by Minkowski inequality, $(iii)$ is by Cauchy–Schwarz inequality and the independence of $\varepsilon$ and $\bX$, $(iv)$ is by Assumption~\ref{assumption:sub1}, Assumption~\ref{assumption:sub2} and the assumption of $\epsilon$, $(v)$ is by $k\geq 1$. Thus, $YS_{ij}(\bX)-\EE[YS_{ij}(\bX)]$ is a sub-exponential random variable and we can get 
\begin{align*}
\|YS_{ij}(\bX)-\EE[YS_{ij}(\bX)]\|_{\psi_1} &= \sup_{k\geq 1}\frac{1}{k}\bigg(\EE|YS_{ij}(\bX)-\EE[YS_{ij}(\bX)]|^k\bigg)^\frac{1}{k}\\
&\leq 4T+\sigma_\varepsilon.
\end{align*}
Then, by Bernstein's inequality, for any $t>0$, we have 
\begin{align*}
\PP\left(\left|\frac{1}{n}\sum_{k=1}^n Y_kS_{ij}(\bX_k)-\EE[YS_{ij}(\bX)]\right|\geq t\right) \leq 2\cdot\mathrm{exp}\left[-cn\min\left(\frac{t^2}{K^2},\frac{t}{K}\right)\right],
\end{align*}
where $K\leq 4T+\sigma_\varepsilon$ and $c$ is a positive constant. Therefore, we can obtain that
\begin{align*}
&\PP\left(\left\|\frac{1}{n}\sum_{k=1}^nY_kS(\bX_k)-\EE(YS(\bX))\right\|_F \geq \sqrt{pd}\cdot t\right)\\
\leq& \sum_{i=1}^{p}\sum_{j=1}^{d} \PP\left(\left|\frac{1}{n}\sum_{k=1}^n Y_kS_{ij}(\bX_k)-\EE[YS_{ij}(\bX)]\right|\geq t\right)\\
\leq& 2pd\cdot\mathrm{exp}\left[-cn\min\left(\frac{t^2}{K^2},\frac{t}{K}\right)\right].
\end{align*}
Letting $t=\frac{4T+\sigma_\varepsilon}{\sqrt{cn}}\cdot\sqrt{\log\frac{2pd}{\delta}}$, we get
\begin{align*}
\PP\left(\left\|\frac{1}{n}\sum_{k=1}^nY_kS(\bX_k)-\EE(YS(\bX))\right\|_F\geq \sqrt{\frac{pd}{cn}} (4T+\sigma_\varepsilon)\sqrt{\log\frac{2pd}{\delta}}\right) \leq \delta.
\end{align*}
Therefore, we can conclude that with probability at least $1-\delta$, 
\begin{align*}
\left\|\frac{1}{n}\sum_{k=1}^nY_kS(\bX_k)-\EE(YS(\bX))\right\|_F = \mathcal{O}\left(\sqrt{\frac{pd}{n}\cdot\log\frac{pd}{\delta}}\right).
\end{align*}
\end{proof}

By Lemma~\ref{thm:subA}, we utilize a modified Davis-Kahan theorem proposed in \citet{o2018random} to prove Theorem~\ref{thm:subbeta}.

\begin{proof}[\textbf{Proof of Theorem~\ref{thm:subbeta}}]
To simplify the notation, we denote $A=\EE[YS(\bX)]$, $A_i=Y_iS(\bX_i)$, and $\hat{A}=\frac{1}{n}\sum_{i=1}^{n}A_i$. Then, by Theorem 4 in \citet{o2018random}, we can get
\begin{align*}
\sin\Theta\big(\hat{\bb},\bb\big) \leq 2 \cdot \frac{\|\hat{A}-A\|_{op}}{\|A\|_{op}}.
\end{align*}
Since $\|\hat{A}-A\|_{op}\leq\|\hat{A}-A\|_{F}$ and $\|A\|_{op}\geq C_1\sqrt{p}$, we can obtain that with probability $1-\delta$,
\begin{align*}
\sin\Theta\big(\hat{\bb},\bb\big) \leq \frac{2}{C_1} (4T+\sigma_\varepsilon)\sqrt{\frac{d}{cn}\log\frac{2pd}{\delta}}.
\end{align*}
We can assume that $\cos\Theta\big(\hat{\bb},\bb\big) = \langle \hat{\bb},\bb \rangle \geq 0$, otherwise we use $-\hat{\bb}$ to replace $\hat{\bb}$. Then,
\begin{align*}
\|\hat{\bb}-\bb\|_2 &= \sqrt{\|\hat{\bb}\|_2^2+\|\bb\|_2^2-2\langle\hat{\bb},\bb\rangle}\\
&= \sqrt{2-2\cos\Theta\big(\hat{\bb},\bb\big)}\\
&\leq \sqrt{2} \sin\Theta\big(\hat{\bb},\bb\big)\\
&\leq \frac{2\sqrt{2}}{C_1} (4T+\sigma_\varepsilon)\sqrt{\frac{d}{cn}\log\frac{2pd}{\delta}},
\end{align*}
Thus, we have $\|\hat{\bb}-\bb\|_2=\mathcal{O}(\sqrt{d/n\cdot\log(pd/\delta)})$.
\end{proof}

\subsection{Proof of Theorem~\ref{thm:truncbeta}}

We first get the following lemma to show the convergence rate of $\frac{1}{n}\sum_{i=1}^{n}\tau(Y_iS(\bX_i))$.

\begin{lemma}\label{thm:truncA}
Suppose the model $\EE(Y_i)=f(\bX_i\bTheta)$ for $i\in[n]$. Under Assumption~\ref{assumption:trunc1} and Assumption~\ref{assumption:sub3}, with probability at least $1-\delta$, letting 
\begin{align*}
\theta=\sqrt{\frac{2\log(2(p+d)/\delta)}{nMpd}},
\end{align*} we have
\begin{align*}
\left\|\frac{1}{n}\sum_{i=1}^{n}\tau(Y_iS(\bX_i))-\EE[YS(\bX)]\right\|_{op} = \mathcal{O}\left(\sqrt{\frac{pd}{n}\cdot\log\frac{p+d}{\delta}}\right).
\end{align*}
\end{lemma}

\begin{proof}[\textbf{Proof of Lemma~\ref{thm:truncA}}]
For any $u\in \RR^{p}$ such that $\|u\|_2=1$, we have that
\begin{align*}
u^{\top}\EE[YS(\bX)(YS(\bX))^{\top}]u &= \EE[Y^2 \cdot u^{\top}S(\bX)S(\bX)^{\top}u]\\
&= \sum_{j=1}^{d} \EE[Y^2 \cdot (S_{\cdot,j}(\bX)^{\top}u)^2]\\
&\leq \sum_{j=1}^{d} \sqrt{\EE(Y^4) \cdot \EE[(S_{\cdot,j}(\bX)^{\top}u)^4]},
\end{align*}
where the last inequality is by Cauchy–Schwarz inequality. And, for any $j\in[d]$,
\begin{align*}
\EE[(S_{\cdot,j}(\bX)^{\top}u)^4] &= \EE\left[\left(\sum_{i=1}^{p} S_{i,j}(\bX)u_i\right)^4\right]\\
&\overset{(i)}{\leq} \EE\left[\left(\sum_{i=1}^{p}S_{i,j}^2(\bX)\right)^2\left(\sum_{i=1}^{p}u_i^2\right)^2\right]\\
&\overset{(ii)}{\leq} \EE\left[p\sum_{i=1}^{p}S_{i,j}^4(\bX)\right]\\
&\overset{(iii)}{\leq} Mp^2,
\end{align*}
where $(i)$ is by Cauchy-Schwarz inequality, $(ii)$ is by Cauchy-Schwarz inequality and $\|u\|_2=1$, $(iii)$ is by $\EE S_{ij}^4(\bX)\leq M$.

Thus, we obtain that $u^{\top}\EE[YS(\bX)(YS(\bX))^{\top}]u \leq Mpd$, which implies that $\|\EE[YS(\bX)(YS(\bX))^{\top}]\|_{op} \leq Mpd$. Similarly, we can get $\|\EE[(YS(\bX))^{\top} YS(\bX)]\|_{op} \leq Mpd$. Therefore, by Corollary 3.1 in \cite{minsker2018sub}, we can get that
\begin{align*}
\PP\left(\left\|\frac{1}{\theta n}\sum_{i=1}^{n}\psi[\theta Y_iS(\bX_i)]-\EE[YS(\bX)]\right\|_{op} \geq \frac{t}{\sqrt{n}}\right) \leq 2(p+d)\mathrm{exp}\left[-\theta t\sqrt{n}+\frac{\theta^2nMpd}{2}\right]
\end{align*}
for any $t>0$ and $\theta>0$.
Setting $\theta=\sqrt{2\log(2(p+d)/\delta)}/\sqrt{nMpd}$ and $t=\sqrt{2Mpd\log(2(p+d)/\delta)}$, we have that
\begin{align*}
\PP\left(\left\|\frac{1}{n}\sum_{i=1}^{n}\tau(Y_iS(\bX_i))-\EE[YS(\bX)]\right\|_{op} \geq \sqrt{\frac{2Mpd}{n}\log\frac{2(p+d)}{\delta}}\right) \leq \delta,
\end{align*}
which also means that with probability at least $1-\delta$,
\begin{align*}
\left\|\frac{1}{n}\sum_{i=1}^{n}\tau(Y_iS(\bX_i))-\EE[YS(\bX)]\right\|_{op} = \mathcal{O}\left(\sqrt{\frac{pd}{n}\cdot\log\frac{p+d}{\delta}}\right).
\end{align*}
\end{proof}

Then, similar to the proof of Theorem~\ref{thm:subbeta}, we can also prove Theorem~\ref{thm:truncbeta} by directly using the modified Davis-Kahan theorem.

\begin{proof}[\textbf{Proof of Theorem~\ref{thm:truncbeta}}]
The proof is the same as that of Theorem~\ref{thm:subbeta}. We can obtain that
\begin{align*}
\|\Tilde{\bb}-\bb\|_2 &\leq \sqrt{2}\sin\Theta(\Tilde{\bb}-\bb)\\
&\leq 2\sqrt{2} \cdot \frac{\|\frac{1}{n}\sum_{i=1}^{n}\tau(Y_iS(\bX_i))-\EE[YS(\bX)]\|_{op}}{\|\EE[YS(\bX)]\|_{op}}\\
&= \mathcal{O}\left(\sqrt{\frac{d}{n}\cdot \log\frac{p+d}{\delta}}\right).
\end{align*}
\end{proof}

\subsection{Proof of Theorem~\ref{thm:multiple_sub}}

Based on Lemma~\ref{thm:subA} and a modified Davis–Kahan–Wedin sine theorem proposed in \citet{o2018random}, we can prove Theorem~\ref{thm:multiple_sub} as follows.
\begin{proof}[\textbf{Proof of Theorem~\ref{thm:multiple_sub}}]
To simplify the notation, we denote $A=\EE[YS(\bX)]$, $A_i=Y_iS(\bX_i)$, and $\hat{A}=\frac{1}{n}\sum_{i=1}^{n}A_i$. According to Lemma~\ref{thm:subA}, we have that with probability at least $1-\delta$, 
\begin{align*}
\|\hat{A}-A\|_F = \mathcal{O}\left(\sqrt{\frac{pd}{n}\cdot\log\frac{pd}{\delta}}\right).
\end{align*}
Then, by Theorem 19 in \citet{o2018random}, we can obtain that
\begin{align*}
\|\hat{\bTheta}\hat{\bO}-\bTheta\|_{F} \leq \frac{2^{3/2}\|\hat{A}-A\|_F}{\sigma_k} = \mathcal{O}\left(\sqrt{\frac{pd}{n}\cdot\log\frac{pd}{\delta}}\right).
\end{align*}
\end{proof}

\subsection{Proof of Theorem~\ref{thm:multiple_trunc}}

Based on Lemma~\ref{thm:truncA} and a modified Davis–Kahan–Wedin sine theorem proposed in \citet{o2018random}, we can prove Theorem~\ref{thm:multiple_trunc} as follows.
\begin{proof}[\textbf{Proof of Theorem~\ref{thm:multiple_trunc}}]
To simplify the notation, we denote $A=\EE[YS(\bX)]$ and $\Tilde{A}=\frac{1}{n}\sum_{i=1}^{n}\tau(Y_iS(\bX_i))$. According to Lemma~\ref{thm:truncA}, we can get that with probability at least $1-\delta$, 
\begin{align*}
\|\Tilde{A}-A\|_F = \mathcal{O}\left(\sqrt{\frac{pd}{n}\cdot\log\frac{p+d}{\delta}}\right).
\end{align*}
Then, by Theorem 19 in \citet{o2018random}, we can obtain that
\begin{align*}
\|\Tilde{\bTheta}\Tilde{\bO}-\bTheta\|_{F} \leq \frac{2^{3/2}\|\Tilde{A}-A\|_F}{\sigma_k} = \mathcal{O}\left(\sqrt{\frac{pd}{n}\cdot\log\frac{p+d}{\delta}}\right).
\end{align*}
\end{proof}

\subsection{Proof of Theorem~\ref{thm:plug-in_beta}}
\begin{lemma}\label{lemma:covariance} (\cite{vershynin2020high}, Covariance estimation).
	Let $\bX$ be an $n\times p$ matrix whose rows $\bx_i$ are independent, mean zero, sub-gaussian random
	vectors in $\mathbb{R}^p$ with the same covariance matrix $\bSigma$. Denote $\hat{\bSigma} = (1/n)\bX^\T\bX$ as the sample covariance matrix. Then for any $t\ge 0$ we have
	\begin{align}
		\left\|\hat{\bSigma} - \bSigma\right\|_{\text{op}} \le \max(\delta, \delta^2), \quad \text{where } \delta=CLt\sqrt{\frac{p}{n}}
	\end{align}
	with probability at least $1 - 2\exp(-t^2p)$. Here $L=\max(K, K^2)$ and $K = \max_i\|\bx_i\|_{\psi_2}$.
\end{lemma}

\begin{lemma}\label{lemma:general_spectrum} (\cite{vershynin2020high})
	Let $\bX$ be an $n\times p$ matrix whose rows $\bx_i$ are independent isotropic random vectors in $\mathbb{R}^p$. Assume that for some $L\ge 0$, it holds almost surely for every $i$ that $\|\bx_i\|_2 \le L\sqrt{p}$. Then for every $t\ge 0$, one has
    \begin{align*}
        \sqrt{n} - tL\sqrt{p} \le \sigma_{\min}(\bX) \le \sigma_{\max}(\bX) \le \sqrt{n} + tL\sqrt{p}
    \end{align*}
    with probability at least $1 - 2p\cdot\exp(-ct^2)$.
\end{lemma}

\begin{lemma}\label{lemma:inverse} (\cite{xu2020perturbation}, Perturbation of inverse)
	For two matrices $\bA, \bB\in \mathbb{R}^{p\times p}$, we have
	\begin{align}
		\left\|\bB^{-1} - \bA^{-1}\right\|_2 \le\|\bA\|_2\|\bB\|_2\left\|\bB - \bA\right\|_2
	\end{align}
\end{lemma}

\begin{lemma}\label{lemma:norm}
	Let $\bx \in \mathbb{R}^p$ be a sub-gaussian random vector with parameter $\sigma$, then with probability at least $1 - t$ for $t \in (0, 1)$
	\begin{align}
		\|\bx\|_2 \le 4\sigma\sqrt{p} + 2\sigma \sqrt{\log(1/t)}
	\end{align}
\end{lemma}

\begin{proof}[\textbf{Proof of Theorem~\ref{thm:plug-in_beta}}]
According to Taylor's theorem, $f(\bX_i\bTheta)$ can be expanded at the point $\ba = \boldsymbol{0}_p$ as follows with a function $h:\mathbb{R}^{p\times R} \rightarrow \mathbb{R}$
\begin{align*}
    f(\bX_i\bTheta) &= \langle\nabla f(\boldsymbol{0}_{p\times R}), \bX_i\bTheta\rangle + h(\bX_i\bTheta)\|\bX_i\bTheta\| \\
    &= \bx_i^\T\sum_{r=1}^R\left(\bb_r \otimes \nabla f(\boldsymbol{0}_{p\times R})_{[:, r]}\right) + h(\bX_i\bTheta)\|\bX_i\bTheta\|,
\end{align*}
with $\lim_{\bX_i\bTheta\rightarrow \boldsymbol{0}} h(\bX_i\bTheta)=0$,
where $\nabla f(\boldsymbol{0}_{p\times R})_{[:, r]}$ is $r$-th column of $\nabla f(\boldsymbol{0}_{p\times R})$.
Let $\bY = (Y_i, \ldots, Y_n)^\T$, $\bX = (\bx_1, \ldots, \bx_n)^\T$ and $\bR = \left(h(\bX_1\bTheta)\|\bX_1\bTheta\|, \ldots, h(\bX_n\bTheta)\|\bX_n\bTheta\|\right)$, then we can write the model in a matrix form below
\begin{align}
    \bY = \bX\sum_{r=1}^R\left(\bb_r \otimes \nabla f(\boldsymbol{0}_{p\times R})_{[:, r]}\right) + \bR + \boldsymbol{\varepsilon}.
\end{align}
Let $\bS = (S(\bx_1), \ldots, S(\bx_n))^\T$ be the full score matrix, it holds for gaussian distribution with zero mean and covariance $\bSigma$ that $\bS = \bX\bSigma^{-1}$.
Therefore, we have the truth-based estimator by algebra
\begin{align}\label{eq:truth-based}
    \text{vec}\left[\frac{1}{n}\sum_{i=1}^n Y_i S(\bX_i)\right]
    = \bSigma^{-1}\hat{\bSigma} \sum_{r=1}^R\left(\bb_r \otimes \nabla f(\boldsymbol{0}_{p\times R})_{[:, r]}\right) + \frac{1}{n}\bSigma^{-1}\bX^\T(\bR + \boldsymbol{\varepsilon}).
\end{align}
If we replace $\bSigma$ by $\hat{\bSigma}$, we can directly get the plug-in estimator below
\begin{align}\label{eq:plug-in}
    \text{vec}\left[\frac{1}{n}\sum_{i=1}^n Y_i \hat{S}(\bX_i)\right] 
    =\sum_{r=1}^R\left(\bb_r \otimes \nabla f(\boldsymbol{0}_{p\times R})_{[:, r]}\right)+ \frac{1}{n}\hat{\bSigma}^{-1}\bX^\T(\bR + \boldsymbol{\varepsilon}).
\end{align}
According to (\ref{eq:truth-based}) and (\ref{eq:plug-in}), we first have
\begin{align*}
	&\left\|\hat{\bA} - \check{\bA}\right\|_F\\
	=& \left\|\text{vec}\left[\frac{1}{n}\sum_{i=1}^n Y_i S(\bX_i)\right] - \text{vec}\left[\frac{1}{n}\sum_{i=1}^n Y_i \hat{S}(\bX_i)\right]\right\|_2 \\
	=&\left\|\left(\bSigma^{-1}\hat{\bSigma} - \bI\right) \sum_{r=1}^R\left(\bb_r \otimes \nabla f(\boldsymbol{0}_{p\times R})_{[:, r]}\right) - \frac{1}{n}\left(\bSigma^{-1} - \hat{\bSigma}^{-1}\right)\bX^\T(\bR + \boldsymbol{\varepsilon})\right\|_2 \\
	\le&\underbrace{\left\|\bSigma^{-1}\hat{\bSigma} - \bI\right\|_2 \left\|\sum_{r=1}^R\left(\bb_r \otimes \nabla f(\boldsymbol{0}_{p\times R})_{[:, r]}\right)\right\|_2}_{A1} + \underbrace{\frac{1}{n}\left\|\bSigma^{-1} - \hat{\bSigma}^{-1}\right\|_2\left\|\bX^\T(\bR + \boldsymbol{\varepsilon})\right\|_2}_{A2}.
\end{align*}
To bound the target error, we shall bound the A1 and A2 terms respectively.
For the first term A1, it holds that
\begin{align*}
	\left\|\bSigma^{-1}\hat{\bSigma} - \bI\right\|_2 = \left\|\bSigma^{-1}\left(\hat{\bSigma} - \bSigma\right)\right\|_2 \le \frac{1}{\lambda_m}\left\|\hat{\bSigma} - \bSigma\right\|_2.
\end{align*}
According to Lemma \ref{lemma:covariance},
\begin{align*}
	\left\|\hat{\bSigma} - \bSigma\right\|_2 \le e, \ \text{where } e=\text{max}(\delta, \delta^2) \text{ and }\delta=CLt\sqrt{\frac{pd}{n}}
\end{align*}
with probability at least $1 - 2\exp(-t^2pd)$. Here $L=\text{max}(K, K^2)$ and $K=M\sqrt{\lambda_M}$. Let $t = \sqrt{\frac{\log(2/\alpha)}{pd}}$, we have
\begin{align*}
	\mathbb{P}\left(\left\|\hat{\bSigma} - \bSigma\right\|_2 \ge CL\sqrt{\frac{\log(2/\alpha)}{n}}\right) \le \alpha,
\end{align*}
which means that
\begin{align}\label{eq:order_1}
	\left\|\hat{\bSigma} - \bSigma\right\|_2 = \mathcal{O}_p\left(\sqrt{1/n}\right).
\end{align}
Besides, note that
\begin{align}\label{eq:order_2}
    \left\|\sum_{r=1}^R\left(\bb_r \otimes \nabla f(\boldsymbol{0}_{p\times R})_{[:, r]}\right)\right\|_2 &= \left\|\text{vec}\left(\nabla f(\boldsymbol{0}_{p\times R})\bTheta^\T\right)\right\|_2 \nonumber \\
    &= \left\|\nabla f(\boldsymbol{0}_{p\times R})\bTheta^\T\right\|_F \nonumber \\
    & = O(\sqrt{p}).
\end{align}
Therefore, we get the order of A1 is $\mathcal{O}_p\left(\sqrt{p/n}\right)$.
For the second term A2, according to Lemma \ref{lemma:inverse}, we first have
\begin{align*}
	\left\|\bSigma^{-1} - \hat{\bSigma}^{-1}\right\|_2 \le \left\|\bSigma^{-1}\right\|_2\left\|\hat{\bSigma}^{-1}\right\|_2\left\|\hat{\bSigma} - \bSigma\right\|_2,
\end{align*}
where for $\left\|\hat{\bSigma}^{-1}\right\|_2$ we have $\left\|\hat{\bSigma}^{-1}\right\|_2 = \lambda_{\max}\left(\hat{\bSigma}^{-1}\right) = 1/\lambda_{\min}\left(\hat{\bSigma}\right) = 1/\sigma^2_{\min}(\bX/\sqrt{n})$. According to Lemma \ref{lemma:general_spectrum}, one has
\begin{align}
	\sigma_{\min}(\bX/\sqrt{n}) \ge 1 - tL\sqrt{pd/n}
\end{align}
with probability $1 - 2pd\exp(-ct^2)$. Let $t = \sqrt{\frac{1}{c}\log\left(\frac{2pd}{\alpha}\right)}$ and we have
\begin{align}
	\mathbb{P}\left(\sigma_{\min}(\bX/\sqrt{n}) \ge 1 - L\sqrt{\frac{pd}{cn}\log\left(\frac{2pd}{\alpha}\right)}\right) \le \alpha.
\end{align}
Therefore we can get the order of $\sigma_{\min}(\bX/\sqrt{n})$ is $\mathcal{O}_p\left(1\right)$, and then
\begin{align}\label{eq:order_emp_pre}
    \left\|\hat{\bSigma}^{-1}\right\|_2=1 / \sigma^2_{\min}(\bX/\sqrt{n}) = \mathcal{O}_p\left(1\right).
\end{align}
Consequently,
\begin{align}\label{eq:order_3}
	\left\|\bSigma^{-1} - \hat{\bSigma}^{-1}\right\|_2 = \mathcal{O}_p\left(\sqrt{1/n}\right).
\end{align}
Besides, for the noise-related term $(1/n)\left\|\bX^\T(\bR + \boldsymbol{\varepsilon})\right\|_2 \le (1/n)\|\bX\|_2\left(\|\bR\|_2 + \|\boldsymbol{\varepsilon}\|_2\right)$.
Again by Lemma \ref{lemma:general_spectrum}, we get $\sigma_{\max}(\bX/\sqrt{n})=\mathcal{O}_p\left(1\right)$, implying that $\|\bX\|_2=\mathcal{O}_p(\sqrt{n})$.
For the remainder vector $\bR$, we have
\begin{align}
    h(\bX_i\bTheta)\|\bX_i\bTheta\| \le H\|\bX_i\|_2.
\end{align}
Besides, according to Lemma \ref{lemma:norm} one has
\begin{align}
    \|\bX_i\|_2 \le \|\bX_i\|_F = \|\bx_i\|_2 = \mathcal{O}_p(\sqrt{pd}).
\end{align}
Therefore
\begin{align}
    \|\bR\|_2 = \sqrt{\sum_{i=1}^n \left(h(\bX_i\bTheta)\|\bX_i\bTheta\|\right)^2} \le H\sqrt{\sum_{i=1}^n\|\bX_i\|_2^2} = \mathcal{O}_p(\sqrt{npd}).
\end{align}
For sub-gaussian $\boldsymbol{\varepsilon}$, according Lemma \ref{lemma:norm} the order of $\boldsymbol{\varepsilon}$ is $\mathcal{O}_p(\sqrt{n})$.
Therefore the order of noise-related term is
\begin{align}\label{eq:order_4}
    \frac{1}{n}\left\|\bX^\T(\bR + \boldsymbol{\varepsilon})\right\|_2 = \mathcal{O}_p(\sqrt{pd}).
\end{align}
Therefore the order of A2 is $\mathcal{O}_p\left(\sqrt{pd/n}\right)$.

In summary, according to (\ref{eq:order_1}), (\ref{eq:order_2}), (\ref{eq:order_3}) and (\ref{eq:order_4}), we can get
\begin{align*}
	&\left\|\hat{\bA} - \check{\bA}\right\|_F = \mathcal{O}_p\left(\sqrt{\frac{pd}{n}}\right).
\end{align*}
Then we can get by Theorem 19 in \citet{o2018random} again
\begin{align*}
\|\check{\bTheta}\check{\bO}-\hat{\bTheta}\|_{F} \leq \sqrt{2}\sin\Theta(\check{\bTheta}, \hat{\bTheta}) \leq \frac{2^{3/2}\|\check{\bA}-\hat{\bA}\|_F}{\sigma_R} = \mathcal{O}_{p}\left(\sqrt{\frac{Rd}{n}}\right).
\end{align*}
This completes the proof.
\end{proof}

\end{document}